\documentclass[5p,final]{elsarticle}
\usepackage{amsfonts}
\usepackage{amssymb}
\usepackage{amsmath}  
\usepackage{amsthm}
\usepackage{enumitem} 
\usepackage{geometry}
\usepackage{graphicx}
\usepackage{algorithm}
\usepackage{algcompatible}
\usepackage{algpseudocode}
\usepackage[colorlinks=true, allcolors=blue]{hyperref}
\usepackage[utf8]{inputenc} 
\usepackage[english]{babel} 
\usepackage{url}
\usepackage{subfiles}
\usepackage{xfrac}
\usepackage{cleveref}
\usepackage{balance}
\usepackage{svg}
\usepackage{booktabs,longtable,array,tabularx,multirow}
\usepackage{tikz}
\usetikzlibrary{shapes.geometric, arrows, arrows.meta, positioning, calc}
\usepackage{rotating}
\usepackage{tikz-3dplot}
\usepackage[squaren]{SIunits}
\usepackage{caption}
\usepackage{natbib}

\usepackage[normalem]{ulem}
\usepackage{subcaption}
\usepackage{float}
\usepackage{placeins}
\usepackage{tocloft}

\newtheorem{theorem}{Theorem}
\newtheorem{lemma}{Lemma}
\newtheorem{definition}{Definition}
\newtheorem{proposition}{Proposition}
\newtheorem{corollary}{Corollary}

\newcommand{\bc}{\begin{center}}
\newcommand{\ec}{\end{center}}
\newcommand{\be}{\begin{equation}}
\newcommand{\ee}{\end{equation}}
\newcommand{\bea}{\begin{eqnarray}}
\newcommand{\eea}{\end{eqnarray}}
\newcommand{\beq}{\begin{eqnarray*}}
\newcommand{\eeq}{\end{eqnarray*}}
\newcommand{\bv}{\left( \begin{array}{c} }
\newcommand{\ev}{\end{array} \right) }

\newcommand{\E}{\mathbb{E}}

\newcommand{\Var}{\mbox{Var}}

\newcommand{\Prob}{\mathbb{P}}

\newcommand{\dd}{\,\mathrm{d}}
\newcommand{\Du}{\Delta u}

\newcommand{\cL}{\mathcal{L}}

\newcommand{\cI}{\mathcal{I}}
\newcommand{\pr}{p_r}

\newcommand{\sigloc}{\sigma_{\mathrm{loc}}}

\newcommand{\Aeff}{\mathcal{A}_{\mathrm{eff}}}
\newcommand{\Fcal}{\mathcal{F}}
\newcommand{\Xiop}{\Xi}

\newtheorem{remark}{Remark}
\newtheorem{assumption}{Assumption}

\begin{document}
\title{Option prices from operational-time reaction-boundary lattices}
\author[unsw-mathstats]{Chris Angstmann}\ead{c.angstmann@unsw.edu.au}
\address[unsw-mathstats]{School of Mathematics and Statistics, University of New South Wales, Sydney, NSW 2052, Australia}
\author[uct-sta]{Tim Gebbie} \ead{tim.gebbie@uct.ac.za}
\address[uct-sta]{Department of Statistical Sciences, University of Cape Town, Rondebosch 7701, Western Cape, South Africa}

\begin{abstract}
We consider the role of a continuum operational time $u$, its mapping to calendar time $t$, and their relation to event time in option-pricing problems. We derive option-pricing equations from an operational-time Markov lattice rather than from a calendar-time diffusion. The primitive model is a homogeneous nearest-neighbour log-price lattice; state-dependent local variance is represented through its general local-kernel extension. Its Chapman--Kolmogorov decomposition yields discrete forward and backward equations. In price variables, the backward equation gives a generalized European pricing PDE and reduces to Black--Scholes--Merton under the risk-neutral drift restriction and constant volatility. Interpreted as a reaction boundary, the lattice gives a structural route from boundary variance to projected local volatility. The key point is the separation of the operational kernel, the calendar-time clock projection, and the pricing-measure choice. Within the deterministic one-factor local-volatility class, an option surface identifies the projected clock--variance combination rather than its operational variance and clock components separately. The resulting hierarchy also distinguishes full calendar-generator equivalence from weaker terminal or continuum matching, and clarifies why incompleteness concerns unspanned clock, jump, or renewal risks rather than random time alone.
\end{abstract}


\maketitle
\tableofcontents

\section{Introduction} \label{sec:intro}

The Black-Scholes-Merton (BSM) equation is usually introduced as a continuous-time pricing equation obtained from geometric Brownian motion, It\^o's lemma, and dynamic replication, and this remains the standard presentation in much of the finance literature \cite{BlackScholes1973,Merton1973}. At the same time, it has long been known that BSM can also be recovered as a limiting case of a discrete-time arbitrage model: \citet{CoxRossRubinstein1979}  explicitly framed their binomial construction as a simple discrete model whose continuous limit contains Black-Scholes as a special case. 

Here we instead start from a nearest-neighbour discrete Markov lattice for the log-price and derive, in a self-contained manner, and without invoking It\^o calculus, both the backward pricing equation for European claims and the associated forward evolution of transition densities. This is not novel in and of itself because what is already known is that BSM admits both continuous-time and discrete-time derivations; what is less standard, and central here, is the explicit use of the discrete Markov kernel itself as the primitive object from which the continuum pricing equations are obtained and interpreted \cite{Kolmogorov1931,Feller1952,BlackScholes1973,Merton1973}. 

In probability theory this is the Kolmogorov--Feller perspective: forward equations evolve densities, while backward equations evolve conditional expectations or contingent claims \cite{Kolmogorov1931,Feller1952}. The discrete arrival equation and the discrete backward valuation equation can be seen as the fundamental, or primitive, adjoint pair from which the continuum forward and backward parabolic equations emerge under the same scaling limit \cite{Kolmogorov1931,Feller1952,BlackScholes1973,Merton1973}. Here this stages the derivation of price and option-pricing models from discrete order-flow models sampled asynchronously in event time.

In finance, the backward equation is the natural vehicle for pricing and hedging a specified payoff because it explicitly discounts payoffs, whereas the forward equation is the natural vehicle for propagating the risk-neutral density and for connecting option prices across strikes and maturities because it is useful for calibration and thus practical hedging. In the local-volatility literature they appear explicitly through the forward density equation and the Dupire relation linking local diffusion to the surface of European option prices. 


A third well-established strand of the literature shows that BSM is not the only possible continuum limit of a discrete or microscopic market model. If one retains a local diffusive scaling with finite variance and weak asymmetry, one obtains geometric Brownian motion or more general local-volatility diffusions; if finite jumps survive, one obtains jump-diffusions or L\'evy-type nonlocal generators; and if random waiting times with heavy tails are retained, one obtains continuous-time random walk (CTRW) or fractional models with memory \cite{Merton1976,ContTankov2004,Clark1973,CarrWu2004,MontrollWeiss1965,MetzlerKlafter2000,ScalasGorenfloMainardi2000,MainardiRabertoGorenfloScalas2000,GorenfloMainardiScalasRaberto2001,CarteaDelCastilloNegrete2007}. Each of these branches is already well known in its own literature: local volatility through Dupire- and Derman--Kani-type arguments \cite{DermanKani1994}, jump and L\'evy extensions through nonlocal pricing operators, and CTRW / fractional finance \cite{ScalasGorenfloMainardi2000, MainardiRabertoGorenfloScalas2000, GorenfloMainardiScalasRaberto2001, MeerschaertScheffler2004, CarteaDelCastilloNegrete2007}. 

Here we start from one discrete Markov/DTRW representation and show in a uniform way how these apparently different pricing regimes arise from different assumptions on the one-step kernel, the waiting-time law, and the limiting scaling procedure \cite{Merton1976,ContTankov2004,MontrollWeiss1965,MetzlerKlafter2000,ScalasGorenfloMainardi2000,MainardiRabertoGorenfloScalas2000,GorenfloMainardiScalasRaberto2001,BaeumerMeerschaert2001,MeerschaertScheffler2004,MeerschaertScheffler2008,CarteaDelCastilloNegrete2007}.


\subsection{Main contributions}
\label{ssec:main_contributions}
The endpoint equations used below are mostly known. The contribution is in how they are put together and what becomes visible when the operational kernel is kept separate from calendar time and from the pricing rule.
\begin{enumerate}[label=(\roman*)]
\item \textbf{Operational-time primitive.} A discrete operational-time transition kernel for a reaction boundary is primitive, rather than a diffusion already specified in calendar time. Forward and backward equations arise first as discrete adjoint relations. \ref{app:generator_convergence} separates the classical expository approach used in the body of the paper with its motivation from process-level generator and martingale-problem convergence.
\item \textbf{Separation of kernel, clock, and pricing measure.} The operational kernel $K_{\Delta u}$, the map $u=U(t)$, and the pricing measure or rule $Q$ are distinct modelling choices. \ref{app:deterministic_clock_projection} proves the deterministic time-only projection and states its boundary relative to state-dependent, direct-subordinator, inverse, and stochastic clocks -- these are heuristically motivated in the body of the text for expository purposes.
\item \textbf{Reaction-boundary interpretation and clock--variance identification.} Local volatility is an activity-rescaled conditional variance rate of the reaction boundary. \ref{app:localvolfromlinearlob} explains how the finite-scale boundary variance derived in the companion paper enters the general local kernel, while \ref{app:clock_variance_equivalence} shows why a deterministic one-factor option surface identifies the projected coefficient rather than its operational variance and clock components separately.
\item \textbf{Unified hierarchy and completeness consequences.} Local diffusion, jump, L\'evy, and renewal or fractional branches sit in one hierarchy. Forward/backward adjointness, terminal-law matching, uniqueness of the pricing measure, and completeness are separate questions. \ref{app:finite_mesh_non_equivalence} gives a finite-mesh option example, while \ref{app:unspanned_clock_risk} shows that randomness alone is not the issue: incompleteness arises when economically relevant clock risk is unspanned or otherwise left unpriced.
\end{enumerate}
\paragraph{Notation.}
The event counter is $n$, operational time is $u$, and calendar time is $t$. Uppercase symbols denote processes or random variables: $S_n$ is the event-time log-price, $S(u)$ the operational-time log-price, and $S_t:=S(U(t))$ its deterministic calendar-time projection, with $X_n=e^{S_n}$, $X(u)=e^{S(u)}$, and $X_t=e^{S_t}$. Lowercase $s$ denotes a generic log-price state and $x=e^s$ the corresponding price state. We write $U(t)$ for a deterministic clock function and $U_t$ for a stochastic clock process.

Under local finite-variance scaling and the convergence conditions stated in the appendices, the lattice processes converge to a diffusion whose generator gives the usual continuum adjoint pair. The remainder of the paper develops these contributions in four steps using classical methods:
\begin{enumerate}[label=(\roman*)]
    \item first, provide a simple template demonstration of the derivation of the option pricing equation from a discrete lattice model starting from first principles;
    \item next, show explicitly how the BSM PDE emerges only in the appropriate diffusion limit;
    \item then, explain how richer discrete kernels generate more general continuous models, including jump and Lévy specifications; and
    \item finally, expand the framework to include fractional continuous-time random walk (CTRW) limits and their implications for option pricing and no-arbitrage.
\end{enumerate}
The discrete framework then naturally generates a hierarchy of continuous limits where Table~\ref{tab:hierarchy} summarizes the main cases. 

The framework is used to show which microscopic assumptions: finite variance, weak bias, locality, jump persistence, or heavy-tailed waiting times; lead respectively to BSM, local-volatility, jump/L\'evy, or fractional alternatives. This places continuous-time pricing equations within a single discrete-to-continuum hierarchy and uses it to clarify how different scaling regimes alter pricing, hedging, no-arbitrage, and market completeness \cite{Kolmogorov1931,Feller1952,BlackScholes1973,Merton1973,Merton1976,ContTankov2004,MontrollWeiss1965,MetzlerKlafter2000,ScalasGorenfloMainardi2000,MainardiRabertoGorenfloScalas2000,GorenfloMainardiScalasRaberto2001,CarteaDelCastilloNegrete2007}. Figure~\ref{fig:compact_dtrw_to_bsm} shows the calendar-time BSM route, while Figure~\ref{fig:dtrw_hierarchy_extended} gives the broader operational-time hierarchy. 

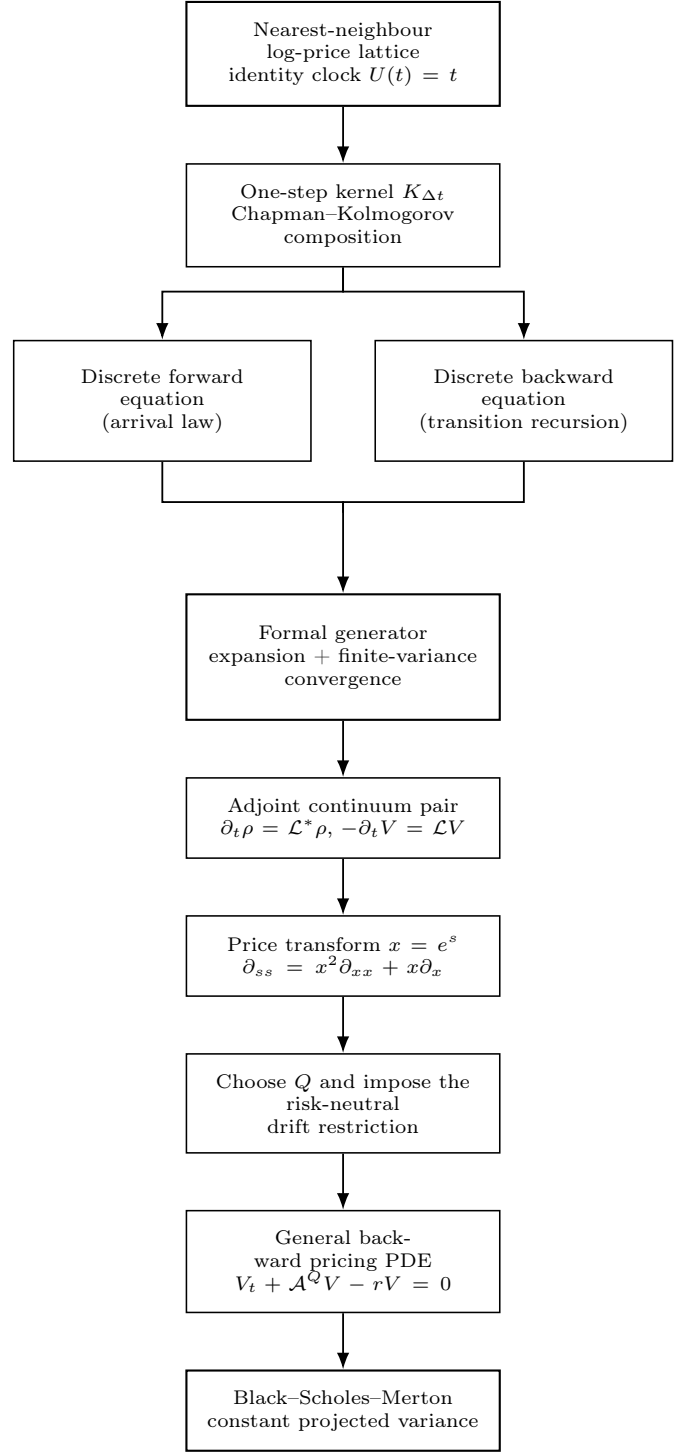
\begin{figure}[htbp]
\centering
\resizebox{\columnwidth}{!}{
\begin{tikzpicture}[
  >=Latex,
  every node/.style={font=\scriptsize},
  box/.style={draw, rectangle, align=center, text width=34mm,
    minimum height=9mm, inner sep=2.5mm, line width=0.55pt},
  key/.style={box, line width=0.75pt},
  pair/.style={box, text width=32mm, minimum height=15mm},
  arrow/.style={->, line width=0.65pt},
  node distance=7mm and 8mm
]
\node[key] (lattice) {Nearest-neighbour log-price lattice\\identity clock $U(t)=t$};
\node[box, below=of lattice] (kernel) {One-step kernel $K_{\Delta t}$\\Chapman--Kolmogorov composition};
\node[pair, below=9mm of kernel, xshift=-22.5mm] (forward)
  {Discrete forward\\equation\\(arrival law)};
\node[pair, below=9mm of kernel, xshift=22.5mm] (backward)
  {Discrete backward\\equation\\(transition recursion)};

\coordinate (pairmid) at ($(forward.south)!0.5!(backward.south)$);
\coordinate (merge) at ($(pairmid)+(0,-6.35mm)$);
\node[key, minimum height=15.62mm, below=10mm of merge] (limit)
  {\mbox{Formal generator}\\expansion + finite-variance\\convergence};
\node[box, below=of limit] (adjoint)
  {Adjoint continuum pair\\$\partial_t\rho=\mathcal L^{*}\rho$, $-\partial_tV=\mathcal LV$};
\node[box, below=of adjoint] (transform)
  {Price transform $x=e^s$\\$\partial_{ss}=x^2\partial_{xx}+x\partial_x$};
\node[box, below=of transform] (pricing)
  {Choose $Q$ and impose the\\risk-neutral drift restriction};
\node[box, below=of pricing] (pde)
  {General backward pricing PDE\\$V_t+\mathcal A^QV-rV=0$};
\node[key, below=of pde] (bsm)
  {Black--Scholes--Merton\\constant projected variance};

\draw[arrow] (lattice) -- (kernel);
\draw[arrow] (kernel.south) -- ++(0,-3mm) -| (forward.north);
\draw[arrow] (kernel.south) -- ++(0,-3mm) -| (backward.north);
\draw[line width=0.65pt] (forward.south) -- ++(0,-5mm) -| (merge);
\draw[line width=0.65pt] (backward.south) -- ++(0,-5mm) -| (merge);
\draw[arrow] (merge) -- (limit);
\draw[arrow] (limit) -- (adjoint);
\draw[arrow] (adjoint) -- (transform);
\draw[arrow] (transform) -- (pricing);
\draw[arrow] (pricing) -- (pde);
\draw[arrow] (pde) -- (bsm);
\end{tikzpicture}}
\caption{Calendar-time identity-clock route from a nearest-neighbour log-price lattice to BSM. The Taylor expansion formally identifies the candidate generator under weak bias, locality, and finite-variance scaling; process-level convergence is stated separately in \ref{app:generator_convergence}. The discrete adjoint equations are developed in Sections~\ref{ssec:forward_operational_discrete}--\ref{ssec:backward_operational_discrete}, and the calendar-time pricing specialization in Proposition~\ref{prop:lattice_to_bsm_rigorous} and Section~\ref{ssec:nearest_neighbour_branch}.}
\label{fig:compact_dtrw_to_bsm}
\end{figure}

\subsection{Forward versus backward equations}
\label{sec:forward_backward_localvol}

The forward equation (see \cref{ssec:forward_operational_discrete}) evolves the transition density (or state-price density, once discounted under a pricing measure) as a function of the future state and future time, holding the current state fixed; by contrast, the backward equation (see \cref{ssec:backward_operational_discrete}) evolves conditional expectations as a function of the current state and current time, holding the terminal payoff fixed. In probabilistic language, the forward equation is the evolution equation for distributions, while the backward equation is the evolution equation for test functions or contingent claims under the infinitesimal generator. 

In the paper's discrete setting this distinction already appears at the lattice level: equation~\eqref{eq:forward_operational_discrete} is an arrival, or forward relation, while equation~\eqref{eq:backward_operational_discrete} is the backward discrete
Kolmogorov relation; after passage to the diffusion limit they become the
adjoint forward and backward parabolic equations
\eqref{eq:forward_operational_limit} and
\eqref{eq:backward_operational_limit} associated with the same local generator: \cite{Kolmogorov1931,Feller1952}.

Briefly, for a time-inhomogeneous Markov generator acting on smooth test functions \(\varphi(x,t)\),
\[
(\mathcal L_t \varphi)(x,t)=b(x,t)\partial_x\varphi(x,t)
+\tfrac12 a(x,t)\partial_{xx}\varphi(x,t),
\]
the backward pricing equation operates on test functions $V(x,t)$:
\[
\partial_t V+\mathcal L_t V-rV=0, \qquad V(x,T)=g(x).
\]
The forward equation acts on densities $p(x,t)$:
\[
\partial_t p=\mathcal L_t^*p
=-\partial_x\{b(x,t)p\}+\tfrac12\partial_{xx}\{a(x,t)p\},
\]
through the adjoint operator \cite{Kolmogorov1931,Feller1952}. 

In finance the distinction is not merely formal. The backward equation is the natural object for pricing a given payoff: starting from the terminal condition $V(x,T)=g(x)$, one propagates the value backward to obtain the current option price: $V(x,t)$. 

\subsection{Forward and backward BSM formulations} \label{ssec:fwdbwkBSM}

In the classical BSM setting the two formulations are, under the standard assumptions, simply two sides of the same model. Under the risk-neutral measure, if the stock follows a geometric Brownian motion, with short rate \(r\) and
continuous dividend yield \(q\) \cite{Merton1973}, then the option value solves the familiar backward parabolic BSM PDE
\begin{equation}
V_t+\frac{1}{2}\sigma^2 x^2 V_{xx}+(r-q)xV_x-rV=0, 
\end{equation}
with terminal condition $V(x,T)=g(x)$. 
Here the discrete lattice first produces a forward relation for the transition probabilities, then an adjoint backward relation, and finally the generalized pricing equation after transformation from log-price to price variables. In this sense, the BSM PDE can be viewed as a backward pricing manifestation of a particular diffusive scaling limit of a discrete Markov kernel. 

The conceptual point is that if one changes the scaling regime, the one-step kernel, or the admissible limiting moments, the backward pricing equation changes because the underlying generator changes. The forward and backward equations remain adjoints only insofar as they continue to arise from the same underlying pricing dynamics \cite{Kolmogorov1931,Feller1952,BlackScholes1973,Merton1973}. All these choices have already been made {\it a-priori} in the conventional continuous-time stochastic analysis approach.

\subsection{Local volatility} \label{ssec:localvolbridge}
The local-volatility case sits naturally between the general lattice discussion and the BSM specialization. In the hierarchy table (See Figure~\ref{fig:dtrw_hierarchy_extended} and Table~\ref{tab:hierarchy}), a state- and time-dependent diffusion limit leads to a general It\^o diffusion
\begin{equation}
\mathrm dX_t=b^Q(X_t,t)\,\mathrm dt+X_t\sigma_{\mathrm{loc}}(X_t,t)\,\mathrm dW_t,
\end{equation}
with backward pricing PDE
\begin{equation}
V_t+b^Q(x,t)V_x+\frac{1}{2}x^2\sigma_{\mathrm{loc}}^2(x,t)V_{xx}-rV=0.
\end{equation}
This includes the BSM model as the special case $b^Q(x,t)=(r-q)x$
and $\sigma_{\mathrm{loc}}(x,t)=\sigma$, but also allows a local-volatility specification in which the relative diffusion coefficient varies deterministically with the current state and time.

The forward equation is then the corresponding Fokker--Planck equation for the risk-neutral density, while the backward equation prices claims under the same local generator. Thus local volatility is not a different pricing principle from BSM; rather, it is a more general state-dependent diffusion realization of the same forward/backward duality \cite{BlackScholes1973,Merton1973,Kolmogorov1931,Feller1952}.

A backward local-volatility PDE tells us how a particular payoff should be valued once the relative-volatility function $\sigma_{\mathrm{loc}}(x,t)$ and the risk-neutral drift are specified. A forward local-volatility equation, in contrast, tells us how the entire risk-neutral density surface evolves. Hence local volatility is especially natural when one wants to link pricing to the cross-section of options, because a single forward density at maturity prices every terminal payoff by integration. 

By contrast, the backward PDE is often the more convenient vehicle for computing one price and its Greeks, especially for non-smooth payoffs or boundary-value formulations. 
Both are consistent when they come from the same local diffusion under the same measure \cite{Kolmogorov1931,Feller1952,BlackScholes1973,Merton1973} given that the same global calendar time is used. 

\subsection{Forward and backward equivalence} \label{ssec:equivalence}
The equivalence between the forward and backward formulations becomes weaker---or can fail outright---once one moves beyond the idealized setting where the two operators are adjoint and the integral pricing representation links them exactly. There are three important ideas here that should not be conflated: First, a Markov generator may still possess a forward/backward adjoint pair, even when it is nonlocal;  Second, the existence of such an adjoint pair does
not by itself imply that the pricing measure is unique;  Third, even under a chosen pricing measure, exact replication may fail if the traded assets do not span all sources of risk. 

Thus what weakens beyond BSM is not
necessarily the algebraic forward/backward duality, but the stronger identification between density propagation, unique valuation, and dynamic hedging. The most important reasons for this are:
\begin{enumerate}
    \item \textbf{Finite mesh / genuinely discrete models.} At non-vanishing $\Delta t,\Delta S$, the correct objects are the discrete forward and backward recursions, not a continuous PDE. Different discretizations with the same first two limiting moments can have materially different finite-mesh option prices, hedging errors, and convergence speeds. \ref{app:finite_mesh_non_equivalence} gives a two-period example in which all European claims at the common maturity agree exactly, while a discretely monitored one-touch claim differs because activity is distributed differently across the two dates. Thus two lattices may share the same continuum limit but still differ at traded maturities or coarse sampling frequencies.

    \item \textbf{Jumps and nonlocal kernels.} Once finite jumps survive the limit, the pricing equation becomes a PIDE rather than a local PDE, and the forward equation carries a nonlocal integral term as well. The forward/backward adjoint relationship still exists at the generator level, but the market is generally no longer complete, and no-arbitrage need not select a unique pricing measure \cite{Merton1976,ContTankov2004}.

    \item \textbf{Fractional CTRW / non-Markovian clocks.} In the CTRW setting with heavy-tailed waiting times, the process can lose the Markov property in calendar time. Then the usual semigroup duality behind the standard forward/backward pair must be reformulated with memory kernels or time-fractional operators, and the classical replication logic behind BSM ceases to apply in its original form. Pricing may still be arbitrage-free under a suitable equivalent pricing measure, but uniqueness is generally lost and incompleteness becomes the natural generic case \cite{MontrollWeiss1965,MetzlerKlafter2000,ScalasGorenfloMainardi2000,MainardiRabertoGorenfloScalas2000,GorenfloMainardiScalasRaberto2001,CarteaDelCastilloNegrete2007}.

    \item \textbf{Additional risk factors.} If the apparent local volatility is only a projection of a higher-dimensional model (for example, one with hidden volatility, rough volatility, liquidity, regime, or jump state variables), then a one-factor forward density may fit European prices while the backward pricing of path-dependent or hedging-sensitive claims differs. In that case ``equivalence'' holds only for the restricted set of payoffs spanned by the projected marginal law, not for the full trading problem. This is one reason why matching one-dimensional terminal distributions need not imply equivalent hedging performance or market completeness. Rough-volatility models make a similar point: that the marginal option surface may be well represented, but the volatility state is a separate rough factor, so hedging requires more than the one-dimensional projected density unless the relevant variance or forward-variance risks are themselves spanned \citep{GatheralJaissonRosenbaum2018,BayerFrizGatheral2016}.
\end{enumerate}

\subsection{Discrete sampling and its implications.} \label{ssec:discretesampling}

Here we emphasize that BSM is a continuum approximation, not the finite lattice itself. We distinguish at least three forms of discreteness: the spacing of the underlying state-time lattice, the monitoring schedule of the payoff, and the rebalancing frequency of the hedge. We then speculate about the role of discrete sampling. These distinctions matter because forward and backward formulations respond differently to coarse sampling.

First, if the lattice is sampled coarsely in time, higher conditional moments of the one-step kernel need not wash out. Then a nominally diffusive model can inherit skewness, kurtosis, or nonlocal effects at option horizons relevant to trading. In the forward equation this appears as departures from a pure local Fokker--Planck form; in the backward equation it appears as higher-order corrections, jump terms, or persistent discretization bias. A coarse discrete model may therefore price short-dated options, digitals, barriers, or deep out-of-the-money claims differently from the continuum local-volatility approximation even when long-maturity vanilla prices look similar.

Second, discrete monitoring of payoffs breaks the naive equivalence between a continuously monitored PDE model and the actual contract specification. This is especially relevant for path-dependent structures, but even for European claims the effective distribution sampled at maturity can depend on whether the underlying dynamics are interpreted on a transaction-time clock, calendar-time clock, or a subordinated/random clock as in CTRW-type models. If the calendar-time dynamics are non-Markovian because of random waiting times, the forward density relevant for contract valuation may differ materially from the one implied by a time-homogeneous diffusion calibrated only at daily sampling \cite{MontrollWeiss1965,MetzlerKlafter2000,ScalasGorenfloMainardi2000,MainardiRabertoGorenfloScalas2000,GorenfloMainardiScalasRaberto2001,CarteaDelCastilloNegrete2007}.

Third, discrete hedging weakens the practical force of backward-PDE replication even in models that are complete in continuous time. In BSM and deterministic local-volatility models, completeness and uniqueness are theoretical statements under frictionless continuous trading. Once trading is discrete, the hedge derived from the backward PDE is only approximate, and the residual P\&L depends on the realized path, local curvature, and the mismatch between the assumed and realized local dynamics. Hence a model may be arbitrage-free and complete in the continuous-time idealization, but still generate significant model risk and hedge slippage under realistic trading schedules. The forward density is then useful for stress testing and scenario aggregation, whereas the backward PDE remains useful for local sensitivities; neither alone captures the full operational risk picture \cite{BlackScholes1973,Merton1973}.

\subsection{What survives beyond BSM?} \label{ssec:noarbitrage}

From the perspective of no-arbitrage (see \cref{ssec:pricing_nonunique_time,sec:operational_hierarchy}): in the BSM diffusion and, more generally, in a one-factor local diffusion with deterministic coefficients, the combination of no-arbitrage and dynamic spanning leads, under ideal assumptions, to a unique risk-neutral pricing rule, so forward and backward formulations are just equivalent computational expressions of the same complete-market model. 

However, once jumps, extra branches, nonlocal kernels, hidden factors, or fractional clocks survive the limit, no-arbitrage generally yields only the existence of at least one equivalent pricing measure, not uniqueness. Then the forward density and the backward pricing equation depend on the chosen measure, and distinct admissible measures can produce distinct option prices for claims not perfectly spanned by traded assets \cite{Merton1976,ContTankov2004,MontrollWeiss1965,MetzlerKlafter2000,ScalasGorenfloMainardi2000,MainardiRabertoGorenfloScalas2000,GorenfloMainardiScalasRaberto2001,CarteaDelCastilloNegrete2007}.

In BSM and deterministic local-volatility diffusions they are equivalent because they are adjoints of the same risk-neutral generator. But this equivalence is a feature of a particular modelling regime---local, Markovian, diffusive, and dynamically spanned---not a universal fact. 

The moment one changes the discrete sampling structure or allows richer lattice limits, the relation between forward law, backward price, and hedge replication becomes more delicate, and the economic conclusions can shift from complete-market uniqueness to incomplete-market choice among admissible pricing measures. This is the value of the discrete-lattice viewpoint: it shows not only how BSM emerges, but also which assumptions are responsible for its apparent universality \cite{Kolmogorov1931,Feller1952,BlackScholes1973,Merton1973,Merton1976,ContTankov2004,MontrollWeiss1965,MetzlerKlafter2000,ScalasGorenfloMainardi2000,MainardiRabertoGorenfloScalas2000,GorenfloMainardiScalasRaberto2001,CarteaDelCastilloNegrete2007}.

\subsection{The derivation in calendar time}

For the identity-clock specialization, the generalized option pricing PDE is derived from a discrete Markov lattice whose nearest-neighbour log-price step is written directly in calendar time $t$. The derivation is visualised in Figure \ref{fig:compact_dtrw_to_bsm}. The BSM PDE is recovered as only one member of a broader hierarchy of continuum limits. Other discrete kernels lead naturally to local-volatility diffusions, jump-diffusions, exponential Lévy models, and fractional CTRW-based dynamics, and so on. As soon as nonlocal jumps, extra state variables, or anomalous waiting times survive the limit, the pricing equation needs to be modified accordingly and this suggests that classical completeness/uniqueness conclusions generally fail.
\begin{proposition}[Calendar-time lattice-to-BSM framework]
\label{prop:lattice_to_bsm_rigorous}
For this proposition only, take the deterministic identity clock $U(t)=t$, so that the lattice step can be written in calendar time as $\Delta t$. The minimal branch is a nearest-neighbour Markov chain $S^{\Delta}_n=\ln X^{\Delta}_n$ on a homogeneous log-price lattice with spacing $\Delta S>0$ and time step $\Delta t>0$: from $(s,t)$ it moves to $s+\Delta S$ with probability $p_r^{\Delta}(s,t)$ and to $s-\Delta S$ otherwise. This branch has constant leading second moment. The state-dependent extension below is understood through a local triangular-array kernel, state-dependent local spacing, or holding intensity with the same limiting first two moments; \ref{app:localvolfromlinearlob} gives an explicit realisation. Let $x=e^s$ denote the corresponding price variable. Assume the following.
\begin{enumerate}[label=(\alph*)]
    \item \textbf{Discrete backward equation.} \label{item:prop1_discrete_backward} For fixed terminal state and time $(s,T)$, the transition density or transition mass satisfies the one-step backward Kolmogorov relation
    \begin{equation}
    \begin{aligned}
    \rho_S^{\Delta}(s,T&\mid s_0,t_0)
    =p_r^{\Delta}(s_0,t_0)\rho_S^{\Delta}(s,T\mid s_0+\Delta S,t_0+\Delta t)
    \\
    &+\bigl[1-p_r^{\Delta}(s_0,t_0)\bigr]\rho_S^{\Delta}(s,T\mid s_0-\Delta S,t_0+\Delta t).
    \end{aligned}
    \label{eq:prop1_discrete_backward}
    \end{equation}
    \item \textbf{Local finite-variance scaling.}  \label{item:prop1_scaling_D_mu}
    There exist locally bounded functions $D(s,t)>0$ and $\mu(s,t)$ such that, locally uniformly in $(s,t)$,
    \begin{align}
    D(s,t)&=\lim_{\Delta t\to0}\frac{1}{2\Delta t}\mathbb E^{\Delta}\!\left[(S^{\Delta}_{n+1}-S^{\Delta}_n)^2\mid S^{\Delta}_n=s\right],
    \\ \nonumber 
    \mu(s,t)&=\lim_{\Delta t\to0}\frac{1}{\Delta t}\mathbb E^{\Delta}\!\left[S^{\Delta}_{n+1}-S^{\Delta}_n\mid S^{\Delta}_n=s\right].
    \label{eq:prop1_scaling_D_mu}
    \end{align}
    For the homogeneous nearest-neighbour lattice these reduce to
    \begin{align}
    D&=\lim_{\Delta S,\Delta t\to0}\frac{(\Delta S)^2}{2\Delta t},
    \\
    \mu(s,t)&=\lim_{\Delta S,\Delta t\to0}\frac{\Delta S}{\Delta t}\bigl(2p_r^{\Delta}(s,t)-1\bigr).
    \label{eq:prop1_nn_scaling_D_mu}
    \end{align}
    Thus the homogeneous fixed-mesh branch has constant $D$. State dependence belongs to the stated local-kernel extension, not to the right-jump probability alone. In the constant-volatility BSM specialization $D=\sigma^2/2$.
    \item \textbf{Regularity and convergence.} For the formal derivation, the limiting coefficients and test functions are sufficiently regular for the Taylor expansion to identify the candidate backward generator
    \begin{equation}
    (\mathcal L_t^S\varphi)(s)=\mu(s,t)\partial_s\varphi(s)+D(s,t)\partial_{ss}\varphi(s).
    \label{eq:prop1_log_generator}
    \end{equation}
    For the process-level statement, assume the local moment convergence, large-increment control, compact containment, and well-posed limiting martingale problem in \ref{app:generator_convergence}. Then the interpolated lattice processes converge weakly to the Markov diffusion associated with the time-inhomogeneous generator $\mathcal L_t^S$. Convergence of transition densities is a stronger local-limit statement and is not needed unless the density representation below is used.
    \item \textbf{Pricing measure and discounting.} A pricing measure $Q$ has been selected on the limiting calendar-time model, the short rate $r$ and dividend yield $q$ are deterministic, and the discounted cum-dividend price is a $Q$-martingale. Within this equivalent local-diffusion change of measure, quadratic variation is unchanged, so $D^Q=D$, whereas the drift is measure dependent. In price coordinates the risk-neutral drift restriction is
    \[
    \mu^Q(s,t)+D^Q(s,t)=r(t)-q(t).
    \]
    For constant $D^Q=D$ this reduces to
        \begin{equation}
        \mu^Q(s,t)=r-q-D.
        \label{eq:prop1_rn_log_drift_D}
        \end{equation}
\end{enumerate}
Before imposing the pricing restriction, the limiting log-price transition law satisfies the backward equation
\begin{align}
-\partial_{t_0}\rho_S(s,T\mid s_0,t_0)
=&\mu(s_0,t_0)\partial_{s_0}\rho_S(s,T\mid s_0,t_0)
\quad \nonumber \\ 
&+D(s_0,t_0)\partial_{s_0s_0}\rho_S(s,T\mid s_0,t_0).
\label{eq:prop1_backward_log_limit}
\end{align}
Under the transformation $x=e^s$, the price-coordinate generator acting on smooth functions $V(x,t)$ is
\begin{align}
\mathcal L^X V(x,t)
=&D(\ln x,t)x^2V_{xx}(x,t) \nonumber \\
&+\bigl[\mu(\ln x,t)+D(\ln x,t)\bigr]xV_x(x,t),
\label{eq:prop1_price_generator_D}
\end{align}
because $\partial_s=x\partial_x$ and $\partial_{ss}=x^2\partial_{xx}+x\partial_x$.
Consequently, for a European payoff $g$ at maturity $T$, define
\begin{equation}
 V(x,t)=\mathbb E^Q\!\left[B_r(t,T)g(X_T)\mid X_t=x\right].
 \label{eq:prop1_european_price_expectation}
\end{equation}
When a transition density exists and $r$ is constant, this is equivalently
\begin{equation}
V(x,t)=e^{-r(T-t)}\int_0^{\infty}g(y)\rho_X(y,T\mid x,t)\,\dd y.
\label{eq:prop1_european_price_integral}
\end{equation}
Under the usual Feynman--Kac regularity or viscosity-solution conditions, the value solves the generalized backward pricing equation
\begin{align}
V_t(x,t)&+D^Q(\ln x,t)x^2V_{xx}(x,t) \nonumber \\
&+\bigl[\mu^Q(\ln x,t)+D^Q(\ln x,t)\bigr]xV_x(x,t)
-rV(x,t)=0,
\label{eq:prop1_general_pricing_D}
\end{align}
where $V(x,T)=g(x)$. If $D^Q=D=\sigma^2/2$ is constant and the risk-neutral log-drift condition \eqref{eq:prop1_rn_log_drift_D} holds, then \eqref{eq:prop1_general_pricing_D} reduces to
\begin{equation}
V_t(x,t)+(r-q)xV_x(x,t)+D x^2V_{xx}(x,t)-rV(x,t)=0,
\label{eq:prop1_bsm_D}
\end{equation}
or, equivalently,
\begin{equation}
V_t(x,t)+(r-q)xV_x(x,t)+\frac{\sigma^2}{2}x^2V_{xx}(x,t)-rV(x,t)=0.
\label{eq:prop1_bsm_sigma}
\end{equation}
This is the Black--Scholes--Merton PDE\footnote{Here, unless otherwise stated, $r$ and $q$ are constant short rate and dividend yield; the deterministic time-dependent cases are $r(t)$ and $q(t)$ with discount factors $B_r(t,T)=\exp({-\int_t^T r(s)ds})$ and $B_q(t,T)=\exp({-\int_t^T q(s)ds)}$. For convenience we often drop the explicit $t$ dependence as these are calendar time by construction unless otherwise stated.}.
\end{proposition}

\begin{remark}[Interpretation]
The proposition is a readable derivation for a \emph{particular scaling regime}, not a substitute for the convergence theorem. \ref{app:generator_convergence} supplies sufficient process-level conditions. At finite lattice spacing, the dynamics and valuation recursion remain discrete. The BSM PDE appears only after the local, finite-variance, weak-bias limit, the calendar-time interpretation, and the risk-neutral restriction are imposed. The theorem proves process convergence; convergence of discounted lattice option values additionally requires convergence of pricing kernels and discount factors and, for unbounded payoffs, suitable uniform-integrability or moment conditions. That stronger pricing-convergence statement is not claimed here.
\end{remark}

\begin{figure*}[p]
\centering
\resizebox{!}{0.6984\textheight}{
\begin{tikzpicture}[
  >=Latex,
  every node/.style={font=\scriptsize},
  common/.style={draw, rectangle, align=center, text width=52mm,
    minimum height=8mm, inner sep=2mm, line width=0.55pt},
  spine/.style={draw, rectangle, align=center, text width=42mm,
    minimum height=8mm, inner sep=2mm, line width=0.65pt},
  pair/.style={draw, rectangle, align=center, text width=36mm,
    minimum height=9mm, inner sep=2.3mm, line width=0.55pt},
  branch/.style={draw, rectangle, align=center, text width=42mm,
    minimum height=9mm, inner sep=2mm, line width=0.55pt},
  head/.style={branch, fill=black!7, font=\scriptsize\bfseries},
  arrow/.style={->, line width=0.65pt},
  node distance=4mm and 8mm
]
\node[common, xshift=24.5mm] (event) {Event-time reaction boundary};
\node[common, fill=black!7, font=\scriptsize\bfseries, below=of event] (lattice) {OPERATIONAL-TIME MARKOV LATTICE\\{\tiny Secs.~\ref{sec:DTRW}--\ref{sec:generalised_operational_lattice}}};
\node[common, below=of lattice] (kernel)
  {Primitive one-step kernel $K_{\Delta u}$\\Chapman--Kolmogorov composition};

\node[pair, below=9mm of kernel, xshift=-22.5mm] (forward)
  {Discrete forward equation\\(arrival law)};
\node[pair, below=9mm of kernel, xshift=22.5mm] (backward)
  {Discrete backward equation\\(transition recursion)};
\coordinate (pairmid) at ($(forward.south)!0.5!(backward.south)$);
\coordinate (scalechoice) at ($(pairmid)+(0,-5.95mm)$);

\node[spine, below=3.05mm of scalechoice] (limit)
  {Local finite-variance limit\\weak bias and locality};
\node[spine, below=of limit] (adjoint)
  {Adjoint continuum pair\\$\partial_u\rho=\mathcal L_u^{*}\rho$, $-\partial_uV=\mathcal L_uV$};
\node[spine, below=of adjoint] (projection)
  {Deterministic clock projection\\and price transform $x=e^s$};
\node[spine, below=of projection] (pricing)
  {Choose $Q$ and impose the\\risk-neutral drift restriction};
\node[spine, below=of pricing] (pde)
  {General backward pricing PDE\\$V_t+\mathcal A^QV-rV=0$};
\node[spine, below=of pde, line width=0.8pt] (bsm)
  {Black--Scholes--Merton\\constant projected variance\\{\tiny Sec.~\ref{ssec:nearest_neighbour_branch}}};

\node[head] (lvh) at ($(limit)+(-98mm,0)$) {LOCAL-VOLATILITY BRANCH};
\node[branch, below=of lvh] (lv1)
  {Calendar-time local variance\\{\tiny $\sigma_{\mathrm{loc}}^2(x,t)=2\alpha(t)D_u^Q(\ln x,U(t))$}};
\node[branch, below=of lv1] (lv2)
  {Operational local diffusion\\$\mathcal L_u=\mu_u\partial_s+D_u\partial_{ss}$};
\node[branch, below=of lv2] (lv3)
  {Forward Fokker--Planck and\\backward pricing PDE};
\node[branch, below=of lv3] (lv4)
  {Option surface identifies the\\projected clock--variance combination\\{\tiny Secs.~\ref{ssec:state_dependent_localvol_branch}, \ref{sec:localvol_from_reaction_boundary}}};

\node[head] (jumph) at ($(limit)+(49mm,0)$) {NONLOCAL JUMP BRANCH};
\node[branch, below=of jumph] (jump1)
  {Surviving finite jumps or a\\general L\'evy-type kernel};
\node[branch, below=of jump1] (jump2)
  {Nonlocal generator with\\risk-neutral jump kernel $\nu^Q$};
\node[branch, below=of jump2] (jump3) {Jump-diffusion or L\'evy PIDE};
\node[branch, below=of jump3] (jump4)
  {Pricing is generally non-unique\\when jump risk is unspanned\\{\tiny Sec.~\ref{ssec:jump_levy_branch_operational}}};

\node[head] (clockh) at ($(pricing)+(-49mm,0)$) {RANDOM-CLOCK BRANCH};
\node[branch, below=of clockh] (clock1)
  {Direct subordinator, inverse clock,\\or heavy-tailed renewal law};
\node[branch, below=of clock1] (clock2)
  {Direct case: subordinated semigroup\\Inverse/renewal case: memory};
\node[branch, below=of clock2] (clock3)
  {Time-changed or fractional\\pricing equation};
\node[branch, below=of clock3] (clock4)
  {Non-uniqueness only when\\claim-relevant clock or renewal\\risk is unspanned\\{\tiny Sec.~\ref{ssec:random_clock_fractional_branch}}};

\draw[arrow] (event) -- (lattice);
\draw[arrow] (lattice) -- (kernel);
\draw[arrow] (kernel.south) -- ++(0,-3mm) -| (forward.north);
\draw[arrow] (kernel.south) -- ++(0,-3mm) -| (backward.north);
\draw[line width=0.65pt] (forward.south) -- ++(0,-5mm) -| (scalechoice);
\draw[line width=0.65pt] (backward.south) -- ++(0,-5mm) -| (scalechoice);
\draw[arrow] (scalechoice) -- (limit);

\draw[arrow] (limit) -- (adjoint);
\draw[arrow] (adjoint) -- (projection);
\draw[arrow] (projection) -- (pricing);
\draw[arrow] (pricing) -- (pde);
\draw[arrow] (pde) -- (bsm);

\draw[arrow] (limit.west) -- (lvh.east);
\draw[arrow] (lvh) -- (lv1);
\draw[arrow] (lv1) -- (lv2);
\draw[arrow] (lv2) -- (lv3);
\draw[arrow] (lv3) -- (lv4);

\draw[arrow] (limit.east) -- (jumph.west);
\draw[arrow] (jumph) -- (jump1);
\draw[arrow] (jump1) -- (jump2);
\draw[arrow] (jump2) -- (jump3);
\draw[arrow] (jump3) -- (jump4);

\coordinate (clockchannel) at (clockh.center |- projection.west);
\draw[line width=0.65pt] (projection.west) -- (clockchannel);
\draw[arrow] (clockchannel) -- (clockh.north);
\draw[arrow] (clockh) -- (clock1);
\draw[arrow] (clock1) -- (clock2);
\draw[arrow] (clock2) -- (clock3);
\draw[arrow] (clock3) -- (clock4);

\end{tikzpicture}}
\caption{Operational-time hierarchy of continuum pricing limits. The central BSM route repeats the spine of Figure~\ref{fig:compact_dtrw_to_bsm}, while the local-volatility, nonlocal-jump, and random-clock alternatives leave at their corresponding modelling choices. The direct-subordinator statement assumes a time-homogeneous operational semigroup; inverse or renewal clocks generally require memory formulations. Across all branches, the pricing rule or admissible measure $Q$ is a distinct modelling choice: forward/backward adjointness, uniqueness of $Q$, and market completeness are separate questions. The event and operational lattice construction is developed in Sections~\ref{sec:DTRW}--\ref{sec:generalised_operational_lattice}; the limiting branches are treated in Sections~\ref{ssec:nearest_neighbour_branch}--\ref{ssec:random_clock_fractional_branch}. Rough volatility lies beyond the scalar local-volatility branch, requiring the operational variance coefficient or forward-variance curve to include an additional stochastic memory state.}
\label{fig:dtrw_hierarchy_extended}
\end{figure*}
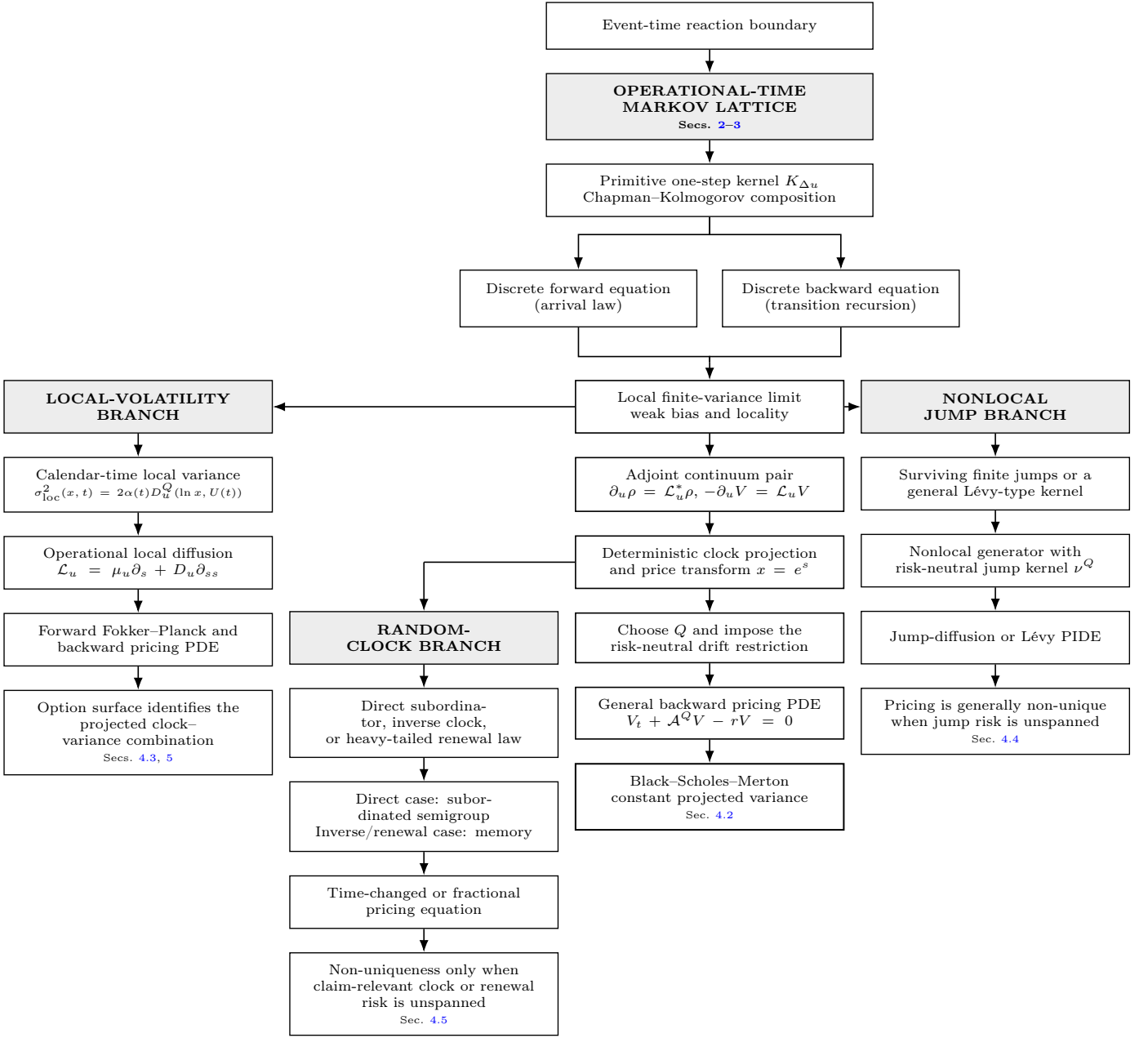

\section{Discrete Event-Time Reaction Boundary}
\label{sec:DTRW}

\subsection{Reaction--diffusion order books in event time}
\label{ssec:reaction_diffusion_event_time}

The primitive object in this construction is not a calendar-time diffusion for the traded price; it is a discrete event system itself. This distinction is important. Let
\[
 n=0,1,2,\ldots
\]
denote the event count, where one event may correspond to an order arrival, cancellation, trade, queue depletion, or any other book event that changes the state relevant for the best bid and best ask. Let $B_n$ and $A_n$ denote the best bid and best ask after event $n$, and let the mid-price be
\begin{equation}
\label{eq:midprice_event}
 M_n:=\frac{A_n+B_n}{2},
 \qquad S_n:=\ln M_n .
\end{equation}
The discrete reaction-boundary interpretation is that the mid-price is the moving interface between event-driven bid-side and ask-side order-flow fields \cite{diana2024anomalous}. If the buy-side field consumes sell-side liquidity faster than the reverse, the boundary tends to move upward; if the sell-side field consumes buy-side liquidity faster, the boundary tends to move downward.

For the first stage we retain only the nearest-neighbour motion of the boundary on a log-price lattice. Let $\Delta S>0$ be the log-price lattice spacing and write
\begin{equation}
\label{eq:event_dtrw_step}
 S_{n+1}=S_n+\Delta S\,\varepsilon_{n+1},
 \qquad \varepsilon_{n+1}\in\{-1,+1\}.
\end{equation}
Here $n$ indexes effective boundary moves. Non-moving book events may instead be absorbed into operational time or represented by holding probabilities in the general kernel.

The event-time transition probabilities are
\begin{align}
\label{eq:event_probabilities}
 \Prob(\varepsilon_{n+1}=+1\mid S_n=s,\mathcal{G}_n)&=p_r(s,n),
 \qquad \text{ and } \\
 \Prob(\varepsilon_{n+1}=-1\mid S_n=s,\mathcal{G}_n)&=1-p_r(s,n),
\end{align}
where $\mathcal{G}_n$ is the event filtration. In the simplest Markov reduction, all relevant order-book information has been projected onto the current boundary location and event index, so that $p_r(s,n)$ is the local right-jump probability of the reaction boundary.

To connect this reduced Markov probability to the order book, introduce a boundary-pressure or imbalance field
\begin{equation}
\label{eq:imbalance}
 \cI(s,n)=\frac{\Lambda_B(s,n)-\Lambda_A(s,n)}{\Lambda_B(s,n)+\Lambda_A(s,n)},
\end{equation}
where $\Lambda_B$ and $\Lambda_A$ are effective bid-side and ask-side reaction intensities near the boundary. A weak-imbalance closure is
\begin{equation}
\label{eq:prob_imbalance}
 2\pr(s,n)-1\simeq \chi\cI(s,n),
\end{equation}
for response coefficient $\chi$. This closure is not needed for the lattice algebra but makes explicit that the jump asymmetry can be interpreted as local order-flow pressure.  This is a finite-mesh closure. Along a diffusion refinement its effective asymmetry must scale as $\Delta u/\Delta S$ to retain a finite drift; \ref{app:scale_clock_rough_caveats} distinguishes this first-moment condition from the boundary-variance calculation. It is a phenomenological closure, not a no-arbitrage restriction.

\begin{remark}[What is being modelled]
The variable $S_n$ is the reaction boundary, not an exogenous geometric Brownian motion. The transition probability $p_r$ is a reduced-form summary of local order-flow imbalance at the boundary. The BSM model will emerge only after a sequence of averaging and limiting operations, and a risk-neutral drift restriction.
\end{remark}

\subsection{Operational time as the continuum limit}
\label{ssec:operational_time_definition}

Event time $n$ is discrete. To take continuum limits while preserving the event-driven interpretation, introduce an operational time $u$ by
\begin{equation}
\label{eq:operational_time_from_events}
 u_n:=n\Delta u,
 \qquad \Delta u>0.
\end{equation}
The continuum operational-time process is the limit of the event-indexed process under the joint refinement
\[
 \Delta S\to0,
 \qquad \Delta u\to0,
 \qquad n\Delta u\to u.
\]
Operational time is therefore the activity clock of the reaction-boundary model in the appropriate limit. It is not, in general, the same object as calendar time $t$ nor the event times $\Delta t_n$ themselves.

Let
\begin{equation}
\label{eq:rho_u_def}
 \rho^{(u)}(s,u\mid s_0,u_0)
\end{equation}
denote the transition mass or density of the log-price boundary in operational time. The superscript $u$ labels the operational-time transition density; it is not a derivative.

\subsection{Discrete forward operational time equation}
\label{ssec:forward_operational_discrete}

The nearest-neighbour one-step kernel in operational time is
\begin{align}
\label{eq:one_step_kernel_operational}
 K_{\Delta u}(s,\mathrm{d}y;u)
 =&p_r(s,u)\,\delta_{s+\Delta S}(\mathrm{d}y) \nonumber \\
 &+\bigl[1-p_r(s,u)\bigr]\,\delta_{s-\Delta S}(\mathrm{d}y).
\end{align}
The Chapman--Kolmogorov relation gives
\begin{equation}
\label{eq:CK_operational}
 \rho^{(u)}(s,u+\Delta u\mid s_0,u_0)
 =\sum_y K_{\Delta u}(y,\{s\};u)\rho^{(u)}(y,u\mid s_0,u_0).
\end{equation}
Since only $y=s-\Delta S$ and $y=s+\Delta S$ can arrive at $s$, this becomes the operational-time arrival equation
\begin{equation}
\label{eq:forward_operational_discrete}
\begin{aligned}
\rho^{(u)}(s,u+\Delta u&\mid s_0,u_0)
=p_r(s-\Delta S,u)\rho^{(u)}(s-\Delta S,u\mid s_0,u_0)
\\
&+\bigl[1-p_r(s+\Delta S,u)\bigr]\rho^{(u)}(s+\Delta S,u\mid s_0,u_0).
\end{aligned}
\end{equation}
The probability multiplying the mass at
$s-\Delta S$ is evaluated here because the chain must jump right from $s-\Delta S$ to arrive at $s$. Similarly, the probability multiplying the mass at
$s+\Delta S$ is evaluated at $s+\Delta S$ because the chain must jump left from
$s+\Delta S$ to arrive at $s$.

\subsection{Discrete backward operational time equation}
\label{ssec:backward_operational_discrete}

Using the first step after $(s_0,u_0)$ gives the adjoint backward relation
\begin{equation}
\label{eq:backward_operational_discrete}
\begin{aligned}
\rho^{(u)}(s,u&\mid s_0,u_0)
=p_r(s_0,u_0)\rho^{(u)}(s,u\mid s_0+\Delta S,u_0+\Delta u)
\\
&+\bigl[1-p_r(s_0,u_0)\bigr]\rho^{(u)}(s,u\mid s_0-\Delta S,u_0+\Delta u).
\end{aligned}
\end{equation}
This is the operational-time version of the discrete backward Kolmogorov equation used later for pricing. 

Unlike the arrival equation, the backward relation conditions on the first step out of the initial state. Hence the probabilities are evaluated at
$(s_0,u_0)$, and the derivatives in the diffusion limit act on the initial
state variable $s_0$, not on the terminal state $s$.

\subsection{Diffusion limit in operational time}
\label{ssec:operational_diffusion_limit}

For the forward equation, to avoid notational ambiguity we temporarily write
\begin{align}
 \rho(s,u):=\rho^{(u)}(s,u\mid s_0,u_0), \text{ and }
  p(s,u):=p_r(s,u). \nonumber 
\end{align}
Define
\begin{equation}
 A(s,u):=p(s,u)\rho(s,u),
\end{equation}
and
\begin{equation}
  B(s,u):=[1-p(s,u)]\rho(s,u).
\end{equation}

Equation~\eqref{eq:forward_operational_discrete} is
\begin{equation}
\label{eq:forward_operational_AB}
 \rho(s,u+\Delta u)=A(s-\Delta S,u)+B(s+\Delta S,u).
\end{equation}
Expanding both sides gives
\begin{align}
\rho(s,u+\Delta u)
&=
\rho
+\Delta u\,\partial_u\rho
+\frac{\Delta u^2}{2}\partial_{uu}\rho
+o(\Delta u^2).
\\
A(s-\Delta S,u)
&=A-\Delta S A_s+\frac{\Delta S^2}{2}A_{ss}+o(\Delta S^2),
\\
B(s+\Delta S,u)
&=B+\Delta S B_s+\frac{\Delta S^2}{2}B_{ss}+o(\Delta S^2).
\end{align}
Since $A+B=\rho$, subtracting $\rho$ and dividing by $\Delta u$ yields
\begin{equation}
\label{eq:operational_prelimit_forward}
 \partial_u \rho
 =\frac{\Delta S}{\Delta u}\partial_s\bigl[(1-2p)\rho\bigr]
 +\frac{\Delta S^2}{2\Delta u}\partial_{ss}\rho
 +o(1),
\end{equation}
This displayed Taylor calculation is the homogeneous nearest-neighbour case with fixed spatial mesh, so the second conditional moment coefficient is constant in \(s\).  The divergence-form forward equation with
state-dependent \(D_u(s,u)\) follows either from a state-dependent local kernel or from a variable activity/mesh scaling, and is stated below in generator-adjoint form. This holds provided the usual weak-asymmetry conditions are satisfied.

Let $\Delta S_u:=S(u+\Delta u)-S(u)$ and assume the operational diffusion scaling
\begin{equation}
\label{eq:operational_diffusion_scaling}
 D_u(s,u):=\lim_{\Delta u\to0}\frac{1}{2\Delta u}
 \mathbb{E}[(\Delta S_u)^2\mid S(u)=s].
\end{equation}
In the homogeneous nearest-neighbour case, this reduces to
$D_u=\lim_{\Delta S,\Delta u\to0}\Delta S^2/(2\Delta u)$.
Assume also the weak-bias scaling
\begin{equation}
\label{eq:operational_drift_scaling}
 \mu_u(s,u):=\lim_{\Delta S,\Delta u\to0}\frac{\Delta S}{\Delta u}\bigl(2p_r(s,u)-1\bigr).
\end{equation}
Then reverting to the prior notation
\begin{align}
\label{eq:forward_operational_limit}
 \partial_u \rho^{(u)} &(s,u\mid s_0,u_0)
 ={}-\partial_s\bigl[\mu_u(s,u)\rho^{(u)} (s,u\mid s_0,u_0)\bigr]
 \nonumber \\
 &+\partial_{ss}\bigl[D_u(s,u)\rho^{(u)} (s,u\mid s_0,u_0)\bigr].
\end{align}
If $D_u$ is constant, this reduces to $-\partial_s(\mu_u\rho^{(u)})+D_u\partial_{ss}\rho^{(u)}$.

The corresponding backward operational-time equation is
\begin{align}
\label{eq:backward_operational_limit}
 -\partial_{u_0}\rho^{(u)}&(s,u\mid s_0,u_0)
={}\mu_u(s_0,u_0)\partial_{s_0}\rho^{(u)} (s,u\mid s_0,u_0) \nonumber \\
&+D_u(s_0,u_0)\partial_{s_0s_0}\rho^{(u)}(s,u\mid s_0,u_0).
\end{align}
For state-dependent $D_u$, the backward generator contains $D_u(s_0,u_0)\partial_{s_0s_0}$, while the forward generator contains the adjoint divergence-form operator $\partial_{ss}(D_u\cdot)$. \ref{app:generator_convergence} separates this formal calculation from process-level convergence.

\subsection{Calendar time as a clock map}
\label{ssec:calendar_clock_map_stage1}

Calendar time $t$ is introduced by a nondecreasing clock
\begin{equation}
\label{eq:clock_map_U}
 u=U(t),
\end{equation}
or equivalently by its inverse $t=T(u)$ when such an inverse exists. The operational-time Markov process is therefore
\[
 S(u),
\]
while the calendar-time observed log-price process is
\begin{equation}
\label{eq:calendar_observed_process}
 S_t:=S(U(t)).
\end{equation}
The corresponding calendar-time price process is
\begin{equation}
\label{eq:operational_calendar_price_processes}
 X_t:=\exp\{S_t\}=\exp\{S(U(t))\}.
\end{equation}
Thus $S(u)$ is indexed by operational time, while $S_t$ and $X_t$ are indexed by calendar time.
Three cases are important.

\begin{enumerate}[label=(\roman*)]
\item \textbf{Deterministic activity clock.} If $U(t)=\alpha t$, then operational and calendar time differ only by scale. For a time-inhomogeneous operational generator, the calendar-time generator is $\alpha\mathcal{L}_{\alpha t}$.
\item \textbf{Time-dependent deterministic activity.} If $\mathrm{d}U(t)=\alpha(t)\mathrm{d}t$, then coefficients are activity-rescaled. For example, an operational variance rate $2D_u$ becomes a calendar variance rate $2\alpha(t)D_u$.
\item \textbf{Random clock.} If $U_t$ is a stochastic clock, then $S_t:=S(U_t)$ is a time-changed process. When the operational dynamics form a time-homogeneous Markov semigroup, a direct independent subordinator gives the classical Bochner-subordinated construction; a time-inhomogeneous model requires an augmented-state or evolution-family formulation. More general random clocks may fail to produce a Markov process in calendar time even when $S(u)$ is Markov in operational time \cite{Bochner1949}.
\end{enumerate}

\begin{proposition}[Operational-time lattice-to-diffusion template]
\label{prop:operational_lattice_template}
Let $S_n$ be the reaction-boundary log-price on the homogeneous nearest-neighbour branch with increments $\pm\Delta S$ in event time. Uppercase $S$ is reserved for the log-price process, while lowercase $s$ denotes a generic log-price state and $x=e^s$ the corresponding positive price state. Let $u_n=n\Delta u$ and let the right-jump probability be $p_r(s,u)$. Suppose the discrete backward equation~\eqref{eq:backward_operational_discrete} holds, the operational variance scaling~\eqref{eq:operational_diffusion_scaling} is finite, and the weak-bias scaling~\eqref{eq:operational_drift_scaling} is finite. The Taylor calculation identifies~\eqref{eq:backward_operational_limit} as the candidate backward equation, with constant $D_u$ for this homogeneous branch. The same conclusion with state-dependent $D_u(s,u)$ holds for a local triangular-array kernel satisfying the moment assumptions of \ref{app:generator_convergence}. If Theorem~\ref{thm:app_lattice_diffusion_convergence} applies, the interpolated lattice converges weakly to the diffusion generated by $\mathcal L_u=\mu_u\partial_s+D_u\partial_{ss}$. Where a sufficiently regular transition density exists, its backward equation is~\eqref{eq:backward_operational_limit}. Calendar-time pricing equations follow only after specifying a clock map $u=U(t)$ and a pricing measure under which the discounted price process is a martingale in calendar time.
\end{proposition}

\begin{remark}[Direct $t$-lattice derivation]
The derivation is unchanged algebraically from the usual lattice-to-diffusion calculation, but its interpretation changes. The Markov kernel and the Chapman--Kolmogorov equations are primitive in $u$, not in $t$. Treating $t$ as primitive implicitly assumes a deterministic identity clock $U(t)=t$. That assumption is precisely what must be relaxed when event time, operational time, and calendar time are not unique.
\end{remark}

\section{Operational-Time Markov Lattice}
\label{sec:generalised_operational_lattice}

\subsection{General one-step kernel in operational time}
\label{ssec:general_kernel_operational}

The nearest-neighbour reaction-boundary model is the minimal local case. The general operational-time Markov lattice is specified by a one-step transition kernel
\begin{equation}
\label{eq:general_kernel_operational}
 K_{\Delta u}(s,\mathrm{d}y;u)
 =\Prob\{S(u+\Delta u)\in\mathrm{d}y\mid S(u)=s\}.
\end{equation}
For a lattice with increments $z\in\mathcal{Z}_{\Delta S}\subset\mathbb{R}$,
\begin{equation}
\label{eq:lattice_jump_kernel}
 K_{\Delta u}(s,\mathrm{d}y;u)
 =\sum_{z\in\mathcal{Z}_{\Delta S}}p_z(s,u)\delta_{s+z}(\mathrm{d}y),
\end{equation}
and
 \begin{equation}
 \sum_z p_z(s,u)=1.
\end{equation}
The nearest-neighbour model corresponds to 
\begin{equation}
    \mathcal{Z}_{\Delta S}=\{-\Delta S,+\Delta S\}.
\end{equation}

For any test function $\varphi$, the one-step backward operator is
\begin{equation}
\label{eq:discrete_backward_operator}
 (P_{\Delta u}\varphi)(s,u)
 :=\int \varphi(y,u+\Delta u)K_{\Delta u}(s,\mathrm{d}y;u).
\end{equation}
The discrete dynamic-programming form for a claim value in operational time is therefore
\begin{equation}
\label{eq:discrete_value_operational_general}
 V(s,u)=e^{-r\Delta t}\int V(y,u+\Delta u)K^Q_{\Delta u}(s,\mathrm{d}y;u),
\end{equation}
where $K^Q$ denotes the pricing-measure kernel and $\Delta t$ is the associated calendar-time increment. This already shows that discounting is a calendar-time operation, while the primitive transition is operational-time. The increment $\Delta t$ is therefore not an additional primitive of the operational lattice. It must be induced by a specified clock rule. For a deterministic absolutely continuous clock $u=U(t)$ with activity rate $\alpha(t)=\mathrm{d}U(t)/\mathrm{d}t$, one has locally $\Delta u\simeq \alpha(t)\Delta t$, so that discounting and coefficient rescaling are both calendar-time operations applied after the operational transition has been specified. If the clock is random or latent, then the pricing measure must also specify the compensation for clock risk; otherwise the backward recursion in operational time does not by itself determine a unique calendar-time price.
The physical/event kernel need not itself be the pricing kernel. Passing from the event-time or operational-time kernel generated by order-flow dynamics to $K^Q$ requires a choice of equivalent pricing measure or another pricing criterion. In the complete local-diffusion case this choice is represented by the drift restriction; outside that case the choice is generally not unique unless the additional risks are traded or otherwise fixed.

\subsection{Operational generator}
\label{ssec:operational_generator}

If the limit exists, the operational-time generator is
\begin{equation}
\label{eq:operational_generator_def}
 \mathcal{L}_u\varphi(s)
 :=\lim_{\Delta u\to0}\frac{1}{\Delta u}
 \left[\int \varphi(y)K_{\Delta u}(s,\mathrm{d}y;u)-\varphi(s)\right].
\end{equation}
For a local finite-variance limit,
\begin{equation}
\label{eq:local_operational_generator}
 \mathcal{L}_u\varphi(s)
 =\mu_u(s,u)\partial_s\varphi(s)+D_u(s,u)\partial_{ss}\varphi(s).
\end{equation}
For a nonlocal jump limit,
\begin{align}
\label{eq:nonlocal_operational_generator}
\mathcal{L}_u\varphi(s)
={}&b_u(s,u)\varphi_s(s)+D_u(s,u)\varphi_{ss}(s) \nonumber \\
&+\int_{\mathbb{R}}\left[\varphi(s+z)-\varphi(s)-z\varphi_s(s)\mathbf{1}_{|z|<1}\right]\nu_u(s,u,\mathrm{d}z).
\end{align}
Here $\nu_u(s,u,\mathrm{d}z)$ is a state- and operational-time dependent L\'evy kernel satisfying the usual local integrability condition
\begin{equation}
\label{eq:levy_kernel_integrability}
 \int_{\mathbb{R}}(1\wedge z^2)\,\nu_u(s,u,\mathrm{d}z)<\infty,
\end{equation}
for each fixed $(s,u)$. This condition is what makes the compensated small-jump term in \eqref{eq:nonlocal_operational_generator} well defined on the usual smooth test-function domain.

Thus BSM, local volatility, jump-diffusion, and L\'evy-type limits differ by the limiting form of $\mathcal{L}_u$.

\subsection{Clock projection to calendar time}
\label{ssec:clock_projection}

Let $U(t)$ be nondecreasing and define $S_t:=S(U(t))$. If $U(t)$ is deterministic and absolutely continuous with
\begin{equation}
\label{eq:activity_rate}
 \frac{\mathrm{d}U(t)}{\mathrm{d}t}=\alpha(t)\ge0,
\end{equation}
then for activity rate $\alpha(t)$ the calendar-time generator is
\begin{equation}
\label{eq:calendar_generator_deterministic_clock}
 \mathcal{A}_t=\alpha(t)\mathcal{L}_{U(t)}.
\end{equation}

\begin{remark}[Boundary of the deterministic clock formula]
Equation~\eqref{eq:calendar_generator_deterministic_clock} applies to a deterministic time-only clock. A state-dependent rate requires a time change built from a path-dependent additive functional and is not represented by simply replacing $\alpha(t)$ with $\alpha(s,t)$ while retaining $U(t)$. Endogenous, coupled, or stochastic clocks require a separate construction; see \ref{app:deterministic_clock_projection}.
\end{remark}
\ref{app:deterministic_clock_projection} gives the corresponding deterministic time-change statement at the martingale-problem and generator levels, including the cases where the activity rate can vanish.

Hence a local operational diffusion becomes
\begin{equation}
\label{eq:calendar_local_diffusion_log}
 \mathrm{d}S_t=\alpha(t)\mu_u(S_t,U(t))\mathrm{d}t
 +\sqrt{2\alpha(t)D_u(S_t,U(t))}\,\mathrm{d}W_t.
\end{equation}
In price variables, $X_t=e^{S_t}$; the calendar-time drift and variance of $X$ acquire the usual geometric transformation.

If $U_t$ is a stochastic clock independent of the operational process, or if an explicit separation or conditional-law construction gives the same factorisation, the transition law is a mixture over operational times:
\begin{equation}
\label{eq:subordinated_density_mixture}
 \rho^{(t)}(s,t\mid s_0,0)
 =\int_0^\infty \rho^{(u)} (s,v\mid s_0,0)h(v,t)\,\mathrm{d}v,
\end{equation}
where $h(v,t)$ is the density of $U_t$ evaluated at operational time $v$. Endogenous or correlated clocks generally require the joint conditional law and need not admit this factorisation. Under the stated restriction the mixture is a bridge from the operational Markov lattice to standard subordinated or CTRW constructions.

\begin{remark}
For a time-homogeneous operational Markov semigroup, a direct independent subordinator gives a Markovian subordinated semigroup. A time-inhomogeneous operational model requires an augmented-state or evolution-family formulation before making the same claim. By contrast, inverse subordinators and renewal clocks, which are typical in CTRW limits with heavy-tailed waiting times, generally produce non-Markovian calendar-time dynamics with memory kernels.
\end{remark}

\subsection{Pricing implication of non-unique time}
\label{ssec:pricing_nonunique_time}

With the projected calendar-time processes denoted by $S_t$ and $X_t=e^{S_t}$, a conventional European payoff at calendar maturity $T$ is $g(X_T)=g(e^{S_T})$. For a deterministic time-only clock, or for a separately specified construction in which the projected process is Markov in the stated variables, its price may be written
\begin{equation}
\label{eq:clocked_pricing_expectation}
 V(x,t)=\mathbb{E}^Q\left[B(t,T)g(X_T)\mid X_t=x\right],
\end{equation}
where $B(t,T)$ is the calendar-time discount factor. For inverse, renewal, endogenous, or correlated clocks, the current price alone need not determine the value. The conditioning state may have to include clock age, current operational time, or another joint price--clock variable; otherwise valuation is a non-Markovian functional. If such clock risk is unspanned, no-arbitrage need not identify a unique pricing measure for it even when the short rate is deterministic.

If the calendar short rate is deterministic and the projected Markov representation applies, then
\begin{equation}
V(x,t)=B(t,T)\,\mathbb E^Q[ g(X_T)\mid X_t=x].
\end{equation}
With stochastic rates or stochastic discounting conventions, the discount factor
must remain inside the conditional expectation. 
The operational-time lattice therefore naturally explains why incompleteness can arise before one adds stochastic volatility or jumps, when only the stock and cash account are traded and clock risk is not itself spanned. \ref{app:unspanned_clock_risk} gives a finite-state construction with a family of equivalent martingale measures that agree on traded stock prices but assign different values to a clock-sensitive claim.

\section{Hierarchy from the Operational Lattice}
\label{sec:operational_hierarchy}

\begin{table*}[p]
\centering
\caption{Hierarchy from operational lattices to calendar-time pricing equations.}
\label{tab:hierarchy}
\small
\renewcommand{\arraystretch}{1.16}
\setlength{\tabcolsep}{3.2pt}
\begin{tabularx}{\textwidth}{
  >{\raggedright\arraybackslash}p{0.165\textwidth}
  >{\raggedright\arraybackslash}p{0.175\textwidth}
  >{\raggedright\arraybackslash}p{0.195\textwidth}
  >{\raggedright\arraybackslash}p{0.205\textwidth}
  >{\raggedright\arraybackslash}X}
\toprule
\textbf{Kernel or clock branch} &
\textbf{Scaling or clock condition} &
\textbf{Calendar-time model} &
\textbf{Pricing and market implication} &
\textbf{Location and Literature} \\
\midrule
Nearest-neighbour log-price lattice &
Local finite variance, weak bias; identity or deterministic clock &
Geometric Brownian motion when projected variance is constant &
Second-order backward PDE; BSM under the risk-neutral drift restriction; complete only under the standard continuous-trading assumptions &
Sections~\ref{ssec:forward_operational_discrete}--\ref{ssec:operational_diffusion_limit}, \ref{ssec:nearest_neighbour_branch}; Proposition~\ref{prop:lattice_to_bsm_rigorous}; \ref{app:generator_convergence}; \citep{CoxRossRubinstein1979,BlackScholes1973,Merton1973}. \\
\addlinespace[4pt]
State/time-dependent local moments &
Local diffusion scaling with $D_u=D_u(s,u)$ and deterministic activity $\alpha(t)$ &
Local-volatility diffusion with $\sigma_{\rm loc}^2(x,t)=2\alpha(t)D_u^Q(\ln x,U(t))$ &
Adjoint Fokker--Planck and backward pricing PDEs; one-factor completeness requires deterministic coefficients and spanning &
Sections~\ref{ssec:state_dependent_localvol_branch}, \ref{ssec:local_moments_boundary}--\ref{ssec:forward_density_dupire_bridge}; \ref{app:localvolfromlinearlob} and~\ref{app:clock_variance_equivalence}; \citep{Dupire1994,DermanKani1994,Gyongy1986}. \\
\addlinespace[4pt]
Multinomial lattice with vanishing increments &
Only the first two scaled moments survive &
Local diffusion; higher finite-mesh cumulants remain approximation corrections &
Second-order generator in the limit; finite-mesh prices and path-sensitive claims may still differ &
Sections~\ref{ssec:general_kernel_operational}--\ref{ssec:operational_generator}; \ref{app:generator_convergence} and~\ref{app:finite_mesh_non_equivalence}; \citep{EthierKurtz1986,StroockVaradhan1979}. \\
\addlinespace[4pt]
Lattice with rare finite jumps &
Jump intensity and finite jump sizes survive &
Jump-diffusion &
PIDE under a selected $Q$; jump risk is generally not spanned by stock and bond alone &
Section~\ref{ssec:jump_levy_branch_operational}; \citep{Merton1976,ContTankov2004}. \\
\addlinespace[4pt]
General nonlocal spatial kernel &
Heavy-tailed or infinite-activity spatial limit &
Exponential L\'evy or stable-type process &
Nonlocal PIDE or fractional-space operator; a risk-neutral kernel $\nu^Q$ and truncation convention must be specified &
Sections~\ref{ssec:operational_generator} and~\ref{ssec:jump_levy_branch_operational}; \citep{ContTankov2004}. \\
\addlinespace[4pt]
Independent direct subordinator &
Independent increasing clock and a time-homogeneous operational semigroup &
Subordinated Markov semigroup &
Calendar dynamics can remain Markovian in the subordinated sense; completeness depends on which clock risks are traded or fixed &
Sections~\ref{ssec:clock_projection}--\ref{ssec:pricing_nonunique_time}; \ref{app:deterministic_clock_projection}; \citep{Bochner1949,Clark1973,CarrWu2004}. \\
\addlinespace[4pt]
Heavy-tailed waiting times or inverse/renewal clock &
Anomalous time scaling with persistent waiting-time effects &
Time-changed process with memory; fractional or memory-kernel evolution &
The classical BSM replication argument does not apply unchanged; non-uniqueness arises when claim-relevant clock or renewal risk is unspanned &
Sections~\ref{ssec:clock_projection}--\ref{ssec:pricing_nonunique_time} and~\ref{ssec:random_clock_fractional_branch}; \ref{app:unspanned_clock_risk}; \citep{MontrollWeiss1965,MetzlerKlafter2000,ScalasGorenfloMainardi2000,MeerschaertScheffler2004}. \\
\bottomrule
\end{tabularx}

\vspace{5pt}
{\footnotesize\itshape Boundary of the scalar hierarchy. Rough volatility requires an enlarged stochastic variance state or forward-variance curve; it is not another one-dimensional nearest-neighbour limit. See Sections~\ref{sec:operational_hierarchy} and~\ref{ssec:calendar_localvol}; \citep{BayerFrizGatheral2016,GatheralJaissonRosenbaum2018}.}
\end{table*}

The hierarchy is given in Table~\ref{tab:hierarchy} and visualised in Figure~\ref{fig:dtrw_hierarchy_extended}. The deterministic local-diffusion branch is derived in detail in this paper. The jump, L\'evy, fractional, rough-volatility, and more general stochastic-clock branches are classifications or extensions requiring their own convergence and pricing constructions. The hierarchy should therefore be read at the level of limiting generators.  When only the first two local cumulants survive, the limit is a local diffusion with a second-order generator. When finite jumps or heavy-tailed spatial increments survive, the limiting generator is nonlocal and the pricing equation is a PIDE or fractional-space equation.  Finite-mesh higher cumulants may produce useful asymptotic corrections, but they should not
be confused with generic finite-order Markov generators beyond second order.  Random clocks also split into cases: direct independent subordination of a time-homogeneous operational semigroup preserves a subordinated Markov semigroup, while inverse or renewal clocks typically produce memory kernels and non-Markovian calendar-time dynamics.  In all non-diffusive branches the displayed pricing equation is conditional on the choice of an admissible pricing measure.  

Rough volatility can be placed in this hierarchy only after enlarging the primitive state. Volterra-type volatility processes can, under suitable structural conditions, be represented through infinite-dimensional Markovian lifts, where generalized Feller semigroup and Kolmogorov-equation methods become available \cite{CuchieroTeichmann2020}. In the scalar lattice used above, the local second conditional moment collapses in the diffusion limit to $D_u(s,u)$, and hence to a deterministic local-volatility coefficient after a deterministic clock projection.

A rough-volatility model instead treats this variance coefficient, or a forward-variance curve generated by it, as an additional stochastic object with low regularity. Thus rough volatility is not a new one-dimensional local diffusion limit of the nearest-neighbour boundary alone; it is an enlarged-state or memory limit in which the operational variance field has its own dynamics \citep{GatheralJaissonRosenbaum2018,BayerFrizGatheral2016}. This keeps rough volatility as a feasible extension of the framework, but not as part of the scalar lattice derivation developed here. 


\subsection{Three primitive modelling choices}
\label{ssec:three_primitives}

The hierarchy is generated by three primitive choices:
\begin{enumerate}[label=(\roman*)]
\item the operational-time one-step kernel $K_{\Delta u}$;
\item the limiting scaling of its conditional moments or cumulants as $\Delta u\to0$;
\item the clock map $u=U(t)$ used to express prices and discounting in calendar time.
\end{enumerate}
The BSM model is obtained only for a local finite-variance kernel, weak asymmetry, a deterministic clock, and a risk-neutral calendar drift. For pricing, there is a fourth choice: the admissible pricing measure $Q$ or pricing criterion used to transform the operational/event kernel into a pricing kernel.  In the complete local-diffusion case this is encoded by the drift restriction; it is otherwise generally an additional model choice.

\subsection{Nearest-neighbour finite-variance branch}
\label{ssec:nearest_neighbour_branch}

For the nearest-neighbour kernel,
\[
 S(u+\Delta u)-S(u)=\pm\Delta S,
\]
with weak asymmetry, the operational generator is
\begin{equation}
\label{eq:nn_operational_generator}
 (\mathcal{L}_u\varphi)(s)=\mu_u(s,u)\varphi_s(s)+D_u(s,u)\varphi_{ss}(s).
\end{equation}
Under a deterministic clock with activity rate $\alpha(t)$, the calendar log-price generator is
\begin{equation}
\label{eq:nn_calendar_generator}
 (\mathcal{A}_t\varphi)(s)=\alpha(t)\mu_u(s,U(t))\varphi_s(s)+\alpha(t)D_u(s,U(t))\varphi_{ss}(s).
\end{equation}
In price variables $x=e^s$, the backward generator acting on smooth functions $V(x,t)$ is
\begin{align}
\label{eq:price_generator_activity}
 \mathcal{A}_t^X V
 ={}&\alpha(t)D_u(\ln x,U(t))x^2V_{xx} \\
 &+\alpha(t)\bigl[D_u(\ln x,U(t))+\mu_u(\ln x,U(t))\bigr]xV_x. \nonumber
\end{align}
The risk-neutral drift condition is
\begin{equation}
    \alpha(t)[D_u^Q(\ln x,U(t))+\mu_u^Q(\ln x,U(t))]=r(t)-q(t)
\end{equation}
Writing $\alpha(t)D_u^Q(\ln x,U(t))=\frac12\sigma_{\mathrm{loc}}^2(x,t)$,
the pricing PDE is then
\begin{equation}
\label{eq:bsm_local_from_activity}
V_t+(r-q)xV_x+\frac12\sigma_{\mathrm{loc}}^2(x,t)x^2V_{xx}-rV=0.
\end{equation}
For $q=0$, this reduces to the zero-dividend BS/local-volatility form. The constant-volatility BSM equation is the special case $\sigma_{\mathrm{loc}}(x,t)\equiv\sigma$.

\subsection{State-dependent local-volatility branch}
\label{ssec:state_dependent_localvol_branch}

If the second operational moment coefficient is state and operational-time dependent,
\begin{equation}
\label{eq:state_dep_second_cumulant}
 D_u(s,u)=\lim_{\Delta u\to0}\frac{1}{2\Delta u}\mathbb{E}\left[(\Delta S_u)^2\mid S(u)=s\right],
\end{equation}
then the corresponding relative local variance in price coordinates is
\begin{equation}
\label{eq:log_local_variance_calendar}
 \sigma_{\mathrm{loc}}^2(x,t)=2\alpha(t)D_u^Q(\ln x,U(t)).
\end{equation}
Equivalently, this is the calendar-time log-variance rate evaluated at $s=\ln x$, while the absolute diffusion coefficient of $X_t$ is $x\sigma_{\mathrm{loc}}(x,t)$.

\subsection{Jump and L\'evy branches}
\label{ssec:jump_levy_branch_operational}

If the operational kernel has rare finite jumps whose intensity survives the limit, then the generator contains a nonlocal term. Under a deterministic clock:
\[
 \nu_t^Q(x,\mathrm dz)
 =\alpha(t)\nu_u^Q(\ln x,U(t),\mathrm dz).
\]
The calendar-time nonlocal term is
\begin{equation}
\label{eq:operational_jump_generator_hierarchy}
 \int_{\mathbb{R}}\left[V(xe^z,t)-V(x,t)-zxV_x(x,t)\mathbf{1}_{|z|<1}\right]\nu^Q_t(x,\mathrm{d}z).
\end{equation}
With the log-price truncation used in \eqref{eq:nonlocal_operational_generator}, the compensation term in price coordinates is $zxV_x\mathbf{1}_{\{|z|<1\}}$. A convention using the price jump $x(e^z-1)$ can also be used, but then the finite-variation drift must be adjusted accordingly. The resulting pricing equation is a PIDE. For each $(x,t)$, assume
$\int_{\mathbb R}(1\wedge z^2)\nu_t^Q(x,\mathrm dz)<\infty$,
so that the compensated nonlocal operator is well defined. The nonlocal operator is not an afterthought: it is the continuum trace of the non-nearest-neighbour operational lattice kernel.

\subsection{Random-clock and fractional branches}
\label{ssec:random_clock_fractional_branch}

If the event waiting times have heavy tails, an inverse clock can generate memory in calendar time. Operational-time Markovianity does not imply calendar-time Markovianity. Under independence, or an explicitly specified separation or conditional-law construction, the density may be represented by a subordination integral of the form
\begin{equation}
\label{eq:hierarchy_subordination_integral}
 \rho^{(t)}(x,t\mid x_0,0)=\int_0^\infty \rho^{(u)}(x,v\mid x_0,0)h(v,t)\,\mathrm{d}v.
\end{equation}
In standard CTRW scaling limits with power-law waiting times, such a representation leads to time-fractional or memory-kernel evolution equations \cite{MagdziarzSchilling2015}. Endogenous or correlated clocks require the joint conditional law and do not generally reduce to the displayed factorisation.
Such models are generally incomplete relative to trading only in the stock and cash account, unless the clock risk is traded or otherwise fixed by an additional pricing criterion. 

\subsection{Boundary to local-volatility}
\label{ssec:boundary_to_lv_conclusion}

The reaction-boundary lattice gives a microscopic interpretation of local volatility:
\begin{equation}
\label{eq:localvol_boundary_summary}
 \sigma_{\mathrm{loc}}^2(x,t)
 =\alpha(t)
 \lim_{\Delta u\to0}\frac{1}{\Delta u}
 \mathbb{E}[(\Delta S_u)^2\mid S(u)=\ln x].
\end{equation}
BSM corresponds to the special case in which this conditional boundary variance is constant and the calendar clock has constant activity. Smile and skew correspond, in this deterministic-clock interpretation, to state-dependent reaction-boundary variance, time-dependent activity, or both.

\section{Local Volatility}
\label{sec:localvol_from_reaction_boundary}

\subsection{Local moments of the boundary motion}
\label{ssec:local_moments_boundary}

Let the operational-time boundary increment be
\[
 \Delta S_u:=S(u+\Delta u)-S(u).
\]
For a general operational kernel $K_{\Delta u}$, define the first local operational drift and the primitive diffusion coefficient by
\begin{align}
 \mu_u(s,u)&:=\lim_{\Delta u\to0}\frac{1}{\Delta u}\mathbb{E}[\Delta S_u\mid S(u)=s],
\label{eq:boundary_first_cumulant}
\\
 D_u(s,u)&:=\lim_{\Delta u\to0}\frac{1}{2\Delta u}\mathbb{E}[(\Delta S_u)^2\mid S(u)=s].
\label{eq:boundary_second_cumulant}
\end{align}
For the nearest-neighbour model,
\begin{align}
\label{eq:nn_boundary_cumulants}
 \mu_u(s,u)&=\lim\frac{\Delta S}{\Delta u}(2p_r(s,u)-1), \nonumber \\
 D_u(s,u)&=\lim\frac{\Delta S^2}{2\Delta u}.
\end{align}
Equivalently, if one writes the second conditional moment rate as $a_u(s,u)$, then $a_u(s,u)=2D_u(s,u)$. Superscript $Q$ denotes coefficients under the selected pricing measure; coefficients without a superscript are generic until a measure is specified, while physical coefficients are written with superscript $P$ when the distinction is needed. Within an equivalent change of measure for this local diffusion, quadratic variation is preserved, so $D_u^Q=D_u^P$; this invariance is not asserted for jump kernels, clock laws, or general lattice measure changes. For the homogeneous fixed-spacing lattice $D_u$ is constant; state dependence belongs to the general local kernel.

\subsection{Operational local-volatility diffusion}
\label{ssec:operational_localvol_diffusion}

The operational-time log-price diffusion is
\begin{equation}
\label{eq:operational_log_localvol_sde}
 \mathrm{d}S(u)=\mu_u(S(u),u)\mathrm{d}u+\sqrt{2D_u(S(u),u)}\mathrm{d}W(u).
\end{equation}
The backward operational generator is
\begin{equation}
\label{eq:operational_localvol_generator}
 (\mathcal{L}_u\varphi)(s)=\mu_u(s,u)\varphi_s(s)+D_u(s,u)\varphi_{ss}(s).
\end{equation}
Under the price-coordinate transformation $x=e^s$, the corresponding operational generator is
\begin{equation}
\label{eq:price_localvol_generator_operational}
 \mathcal{L}_u^X V
 =D_u(\ln x,u)x^2V_{xx}
 +\left[\mu_u(\ln x,u)+D_u(\ln x,u)\right]xV_x.
\end{equation}

\subsection{Calendar-time local volatility}
\label{ssec:calendar_localvol}

With a deterministic activity clock $\mathrm{d}U(t)=\alpha(t)\mathrm{d}t$, the calendar-time relative local variance, equivalently the log-variance rate, is
\begin{equation}
\label{eq:calendar_log_localvar}\sigma_{\mathrm{loc}}^2(x,t)=2\alpha(t)D_u^Q(\ln x,U(t)).
\end{equation}
The corresponding risk-neutral drift condition is
\begin{equation}
\alpha(t)\left[
 \mu_u^Q(\ln x,U(t))+D_u^Q(\ln x,U(t))
\right]=r(t)-q(t).
\end{equation}
The term $D_u^Q$ appears because the martingale condition is imposed on the
price level $x=\exp(s)$, not directly on the log-price state $s$.
The price-coordinate risk-neutral local-volatility PDE is therefore
\begin{equation}
\label{eq:calendar_localvol_pde}
 V_t+(r-q)xV_x+\frac12\sigma_{\mathrm{loc}}^2(x,t)x^2V_{xx}-rV=0,
\end{equation}
where $V(x,T)=g(x)$.
The risk-neutral drift condition fixes the calendar drift of $X_t$ to $(r-q)x$; it does not fix the local variance. The local variance remains twice the activity-rescaled operational diffusion coefficient, equivalently the activity-rescaled second conditional moment rate
under the chosen pricing measure.

\begin{remark}
The option surface identifies the calendar-time local variance $2\alpha(t)D_u^Q(\ln x,U(t))$. Without an additional specification of the activity clock $\alpha(t)$, it does not separately identify the operational variance $D_u^Q$ and the clock activity. \ref{app:clock_variance_equivalence} states this as a calendar-generator equivalence result and describes the resulting equivalence class of deterministic operational representations.
\end{remark}

This also marks the difference between local volatility and rough volatility. In the local-volatility branch, \(D_u^Q(\ln x,U(t))\) is a deterministic coefficient 
under the chosen pricing measure once \(x\) and \(t\) are fixed. In a rough-volatility extension, the corresponding variance coefficient would itself be driven by an additional rough stochastic state, often modelled through a fractional-type process with small Hurst exponent. The Dupire surface can still identify a projected risk-neutral local variance, but it does not by itself identify the dynamics of that rough variance factor or the associated volatility-risk premia 
\citep{GatheralJaissonRosenbaum2018,BayerFrizGatheral2016}.

\subsection{Forward density and Dupire interpretation}
\label{ssec:forward_density_dupire_bridge}

Let $q_X(x,t)$ denote the risk-neutral density of $X_t$. The forward equation
corresponding to \eqref{eq:calendar_localvol_pde} is
\begin{align}
\label{eq:forward_localvol_density}
\partial_t q_X(x,t)
= &-\partial_x\{(r(t)-q(t))x q_X(x,t)\} \nonumber \\
&+\frac12\partial_{xx}\{
\sigma_{\mathrm{loc}}^2(x,t)x^2q_X(x,t)
\}.
\end{align}
This is the Fokker--Planck equation for the same risk-neutral generator whose
backward action gives the pricing PDE. The local-volatility surface is the coefficient that propagates the entire risk-neutral density surface. In the present construction, it is linked directly to the operational reaction-boundary variance under the pricing measure by~\eqref{eq:calendar_log_localvar}.

Concretely, for constant $r$ and $q$, European call prices $C(K,T)$ satisfy the usual Dupire inversion
\begin{equation}
\label{eq:dupire_type_relation}
\sigma_{\mathrm{loc}}^2(K,T)
=\frac{2[C_T+(r-q)KC_K+qC]}{K^2C_{KK}},
\end{equation}
provided the call surface is sufficiently smooth and free of static arbitrage. With deterministic time-dependent rates and dividend yields, the corresponding
formula uses the appropriate instantaneous values and discount-factor convention. This Dupire strike--maturity forward equation is distinct from the
Fokker--Planck forward density equation above, although the two are connected through the Breeden--Litzenberger density relation \cite{BreedenLitzenberger1978}. 

\subsection{Put--call parity} \label{ssec:putcalparity}

For European calls and puts with the same strike, calendar maturity, traded terminal underlying, and settlement convention, put--call parity follows from the pathwise payoff identity. If $\Pi_t$ is their common linear pricing functional, then
\begin{equation}
C-P= \Pi_t(X_T)-K\Pi_t(1).
\label{eq:abstract_parity_pricing_functional}
\end{equation}
When $X_T=e^{S_T}$ is the traded ex-dividend stock delivered at calendar time $T$ and $Q$ is an equivalent martingale measure for its discounted cum-dividend price, no-arbitrage requires
\begin{equation}
\Pi_t(1)=B_r(t,T),
\qquad
\Pi_t(X_T)=xB_q(t,T).
\end{equation}
Hence
\begin{equation}
 C-P=xB_q(t,T)-KB_r(t,T),
 \label{eq:standard_put_call_parity}
\end{equation}
under every admissible martingale measure, even when unspanned clock risk makes the market incomplete. Clock risk can make individual nonlinear claims non-unique across admissible measures, but it cannot change the value of their replicable stock--bond difference.

If an operationally sampled quantity $e^{S(\Theta_T)}$, for a contract-specific operational time $\Theta_T$ not equal to the market clock value $U(T)$, is not the traded stock $X_T=e^{S_T}$ delivered at calendar maturity, the claim contains settlement or basis risk. Its pricing functional, clock law, and carry convention must then be specified before comparison with standard equity-option parity. This is a different contract specification, not a failure of put--call parity. Calendar and operational discount-rate conversion for a deterministic clock is recorded in \ref{app:deterministic_clock_projection}.

\begin{remark}[Dupire put--call parity consistency]
Applying the Dupire strike--maturity forward equation to
$D(K,T)=C(K,T)-P(K,T)$ gives the affine solution
$D(K,T)=x_0B_q(0,T)-KB_r(0,T)$. For constant r and q, this reduces to $D(K,T)=x_0e^{-qT}-Ke^{-rT}$. This is consistent with the standard Dupire/Breeden--Litzenberger density
interpretation \citep{Dupire1994,BreedenLitzenberger1978,Gatheral2006}.
\end{remark}

\section{Discussion}

In the usual BSM construction, Brownian motion is postulated in calendar time and It\^o calculus then yields the pricing equation. Here the order is reversed. Sections~\ref{sec:DTRW}--\ref{sec:generalised_operational_lattice} begin with an event-driven reaction-boundary lattice in operational time, where the Markov transition kernel is primitive and the forward/backward equations are first discrete adjoint relations. Only after the diffusion limit in Section~\ref{ssec:operational_diffusion_limit}, and the clock projection in Sections~\ref{ssec:calendar_clock_map_stage1} and~\ref{ssec:clock_projection}, does the familiar continuous-time diffusion language appear. 

BSM appears in Section~\ref{ssec:nearest_neighbour_branch} only as the local, finite-variance, weak-bias, deterministic-clock case with a risk-neutral drift restriction. Local volatility as it appears in Sections~\ref{ssec:state_dependent_localvol_branch} and~\ref{sec:localvol_from_reaction_boundary} is then the clearest bridge to classical approaches to option pricing \eqref{eq:state_dep_second_cumulant}:
\[
 \sigma_{\mathrm{loc}}^2(x,t)
 =\alpha(t)\lim_{\Du\to0}\frac{1}{\Du}
 \E^Q[(S(u+\Du)-S(u))^2\mid S(u)=\ln x].
\]
Here local volatility is interpreted as the activity-rescaled second conditional moment rate, equivalently the leading local quadratic-variation coefficient, of the reaction-boundary motion under the pricing measure \eqref{eq:log_local_variance_calendar}:
 \[
 \sigma_{\mathrm{loc}}^2(x,t)
 = 2\alpha(t)D_u^Q(\ln x,U(t)).
 \]
 A Dupire surface can then be interpreted as an observational inversion of this risk-neutral boundary variance, not as a primitive microscopic input. More precisely, the inversion recovers the projected calendar-time coefficient under the regular local-volatility assumptions; it does not identify a unique factorisation into operational variance and clock activity. \ref{app:clock_variance_equivalence} makes this distinction explicit. 

Rough volatility sits just outside this scalar local-volatility closure. It would arise if the boundary variance coefficient were not only state dependent but itself a rough stochastic field, or if the projected activity/variance process retained fractional-type memory under calendar projection. In that case the Dupire coefficient is a marginal projection of a richer variance dynamics, and the pricing problem again depends on which variance or forward-variance risks are traded, fixed, or priced by an additional criterion \citep{GatheralJaissonRosenbaum2018,BayerFrizGatheral2016}.

Jump, L\'evy, and fractional alternatives arise when the operational kernel or clock is changed, as in Sections~\ref{ssec:jump_levy_branch_operational} and~\ref{ssec:random_clock_fractional_branch}. The limitation is that completeness, uniqueness of the pricing measure, and exact replication are no longer automatic consequences of an assumed Brownian calendar-time filtration. They must be recovered, or shown to fail, after specifying the limiting generator, the clock, and the traded sources of risk.


For completeness: \ref{app:generator_convergence} supplies the process-level lattice limit; \ref{app:deterministic_clock_projection} proves deterministic clock projection; \ref{app:clock_variance_equivalence} identifies deterministic representations with a common calendar generator; \ref{app:finite_mesh_non_equivalence} shows that weaker terminal matching does not fix path-dependent values; and \ref{app:unspanned_clock_risk} isolates pricing non-uniqueness when a claim-relevant clock state is not spanned by traded payoffs.

\section{Conclusion}

The endpoint equations and stochastic-time constructions are known and operational time itself is not new. What is distinctive here is that the BSM/local-volatility lattice derivation is formulated first in $u$, before calendar projection, and can then be connected to deterministic clocks, direct subordinators, and inverse or renewal clocks. The lattice state is not an abstract stock node but the log of a mid-price interface between bid-side and ask-side order-flow fields. The finite-mesh closure~\eqref{eq:prob_imbalance} links jump asymmetry to boundary pressure; Appendices~\ref{app:localvolfromlinearlob}--\ref{app:scale_clock_rough_caveats} give the corresponding variance bridge and its scale qualifications.

This is because the operational time determines the transition kernel $K_{\Delta u}$ but the calendar time the maturity and discounting; our approach can then be summarised as:
\begin{equation}
    K_{\Delta u} \xrightarrow{U(t)} {\mathcal A}_t \xrightarrow{Q} \mbox{pricing equations} 
\end{equation}
Clock choice $U(t)$ can then change the calendar-time model class as it sets the calendar-time log-price generator ${\mathcal A}_t$: deterministic activity gives local PDEs, direct subordination of a time-homogeneous operational semigroup can give nonlocal Markov generators, and inverse or renewal clocks can give memory equations under their respective standard constructions. Incompleteness arises when clock, jump, or renewal risk is economically relevant and unspanned, not merely because a stochastic clock is present.

This implies that market completeness depends on the projection \cite{AngstmannGebbie2026NonUniqueTime}. In the deterministic local-diffusion case, continuous trading may span the single source of risk. In jump, hidden-factor, random-clock, or renewal-risk models, no-arbitrage generally identifies a family of equivalent pricing measures unless the additional risks are traded or otherwise fixed. Random clocks do not automatically imply incompleteness; unspanned clock or renewal risk does. \ref{app:unspanned_clock_risk} makes the distinction explicit: an extra clock label creates non-uniqueness only when traded payoffs fail to distinguish that label.

The four contributions can be stated plainly. The operational kernel is primitive, with process-level convergence requiring more than a Taylor expansion. Kernel, deterministic clock projection, and pricing rule are separate operations. The option surface identifies the projected clock--variance combination rather than a unique operational factorisation. Finally, terminal-law agreement, path-law agreement, and market completeness are different statements: finite-mesh path claims can separate terminally equivalent models, while unspanned clock states generate non-unique pricing measures. These points are usually distributed across separate literature and hidden once the calendar-time model is taken as given {\it a priori}.

\section*{Acknowledgements}
The authors thank members of the research community for comments that improved the presentation.
\section*{Disclosure of interest}
The authors declare no competing interests.

\section*{Funding}
No funding was received for this research.

\section*{CRediT authorship contribution statement}
Chris Angstmann: Conceptualization, Methodology, Formal analysis, Writing -- review \& editing.
Tim Gebbie: Conceptualization, Methodology, Formal analysis, Writing -- original draft, Writing -- review \& editing.

\section*{AI disclosure}
The authors used generative AI tools to assist with language editing, checking and formatting of the manuscript.
\bibliographystyle{elsarticle-harv}
\bibliography{latticebsm-v9.3}
\appendix

\section{Reaction-boundary variance}
\label{app:localvolfromlinearlob}

This appendix provides a bridge from the displaced latent-book boundary to a finite-scale operational variance coefficient. The Green-function and finite-cutoff spectral calculation is developed in the companion paper \citep{AngstmannGebbie2026AdjointProjection}; only the part needed by the pricing lattice is recorded here.

Let $S(u)$ denote the operational-time log-price boundary. At response scale $\Delta$, define
\begin{equation}
 a_u^{(\Delta)}(s,u)=\frac{1}{\Delta}
 \Var\left[S(u+\Delta)-S(u)\mid S(u)=s\right].
 \label{eq:coarse-rate}
\end{equation}
The boundary is the zero of a bid--ask imbalance field, $\Phi(y(u),u)=0$. Near a simple zero, the latent book is replaced over the response window by its tangent, $\Phi^*(x,u)\simeq-\cL_u(x-y(u))$, where $\cL_u>0$ is the local slope. Filtering signed forcing through the resulting resilient response and taking its second cumulant gives, on the long-memory branch,
\begin{equation}
 \Xiop(s,u;\Delta)=
 \frac{\Aeff(s,u)}{\cL_u(s,u)^2}
 \Delta^{-\gamma(s,u)}
 \Fcal_{\gamma(s,u)}\!\left(\nu_u(s,u)\Delta\right),
 \label{eq:Xi}
\end{equation}
with $0<\gamma(s,u)<1$. Here $\Aeff$ is the centred signed-forcing scale, $\cL_u^{-2}$ is the inverse-square liquidity channel, and $\Fcal_\gamma(\nu_u\Delta)$ is the resilience response. Equation~\eqref{eq:Xi} is the zero-cutoff form; the finite-cutoff expression and its regime conditions remain in the companion paper. Under this long-memory closure, $a_u^{(\Delta)}\simeq\Xiop$ at the chosen response scale.

This finite-scale coefficient can be used in the general local kernel of Section~\ref{ssec:general_kernel_operational}. To see this directly, let $h\downarrow0$ and use the two-point local step
\begin{align}
 \delta_h(s,u)&=\sqrt{a_u^{(\Delta)}(s,u)h},\nonumber\\
 p_h(s,u)&=\frac12+\frac{\mu_u(s,u)h}{2\delta_h(s,u)}.
 \label{eq:local_chain_realisation}
\end{align}
For locally bounded $\mu_u$ and strictly positive $a_u^{(\Delta)}$, this is a valid probability for sufficiently small $h$, and
\begin{align}
 \E[\Delta S_h\mid s,u]&=\mu_u(s,u)h,\nonumber\\
 \E[(\Delta S_h)^2\mid s,u]&=a_u^{(\Delta)}(s,u)h.
 \label{eq:local_chain_moments}
\end{align}
\begin{remark}[Pointwise and compact-uniform validity]
Strict positivity of $a_u^{(\Delta)}(s,u)$ gives pointwise validity of the two-point probability for sufficiently small $h$. Uniform validity on a compact state--time set additionally requires local uniform control, for example a positive lower bound on $a_u^{(\Delta)}$ together with locally bounded $\mu_u$, or a compact bound on $|\mu_u|/\sqrt{a_u^{(\Delta)}}$.
\end{remark}
Thus the limiting second-order generator coefficient is $a_u^{(\Delta)}/2$. The homogeneous fixed-spacing lattice used for the BSM derivation has constant leading variance; state dependence is supplied by this local-kernel extension, not by $p_r$ alone. A common spatial grid could equivalently use holding probabilities or any local kernel with the same limiting moments. This is a Markov closure at fixed mesoscopic scale $\Delta$, not a claim that the boundary response is diffusive below its microstructural cutoff.

There is no conflict between the centred variance in equation~\eqref{eq:coarse-rate} and the raw second moment used in the generator. Since $\E[(\Delta S_h)^2\mid s,u]=\Var(\Delta S_h\mid s,u)+\E[\Delta S_h\mid s,u]^2$, their rates differ by $O(h)$ when $\E[\Delta S_h\mid s,u]=\mu_u h+o(h)$. The superscript $Q$ below is also substantive: the structural calculation supplies a functional form, but does not identify physical order-flow statistics with risk-neutral coefficients.\footnote{The diffusion parameter used in the companion latent-book response is an order-book coefficient and should not be confused with the generator coefficient $D_u$ in Section~\ref{ssec:local_moments_boundary}.}

For a deterministic activity clock $u=U(t)$, the calendar-time local variance is therefore
\begin{align}
 \sigloc^2(x,t)&=\alpha(t)a_u^{(\Delta),Q}(\ln x,U(t)),\nonumber\\
 \alpha(t)&=\frac{\dd U(t)}{\dd t}.
 \label{eq:deterministic-localvol}
\end{align}
The clock determines how quickly operational boundary variance is sampled in calendar time.

\section{Scale, clock, and rough-memory}
\label{app:scale_clock_rough_caveats}

Equation~\eqref{eq:Xi} is a hierarchy statement, not a calibrated volatility-surface model or a complete nonlinear latent-order-book theory. It shows how a finite-scale reaction-boundary variance can enter the general local kernel and then be projected into a local-volatility layer. Four qualifications matter.

First, the imbalance closure and the variance calculation concern different conditional moments. For a two-point step $\pm\delta_h$, finite drift requires
\begin{equation}
 2p_h(s,u)-1=\frac{\mu_u(s,u)h}{\delta_h(s,u)}
 +o\!\left(\frac{h}{\delta_h(s,u)}\right).
 \label{eq:app_weak_bias_scaling}
\end{equation}
Under diffusive scaling $h\asymp\delta_h^2$, the directional asymmetry is $O(\delta_h)$. Thus equation~\eqref{eq:prob_imbalance} is a finite-mesh relation whose response or normalization must scale accordingly along a triangular array. This restricts the first moment only. The order-one covariance amplitude of centred signed forcing, from which equation~\eqref{eq:Xi} is obtained, is compatible with weak directional bias.

Second, the book slope, resilience, forcing scale, and memory exponent are frozen over the response window. The construction is local to a simple boundary zero and does not assert that the full latent book remains linear during a liquidity crisis or across long calendar horizons.

Third, the event mesh, the Markov refinement step $h$, the response window $\Delta$, and calendar maturity are distinct. The compact response factor in equation~\eqref{eq:Xi} additionally assumes separation from the microstructural cutoff; when that separation fails, the finite-cutoff expression should be used.

Fourth, equation~\eqref{eq:deterministic-localvol} is a scalar projection only for a deterministic activity clock. With stochastic, filtered, or non-unique clocks, the forward and backward operators must still arise from the same projected state-price mechanism, and the one-dimensional state space may be insufficient \citep{AngstmannGebbie2026AdjointProjection}.

Rough volatility likewise sits above the scalar local-volatility branch. Variance or the forward-variance curve may be promoted to a memory-bearing state, but under the usual Markovian-projection conditions the Dupire coefficient is a conditional projection such as $\sigma_{\mathrm{Dup}}^2(x,t)=\E^Q[v_t\mid X_t=x]$, not the primitive rough dynamics \citep{Gyongy1986,GatheralJaissonRosenbaum2018,BayerFrizGatheral2016}.

\section{Generator convergence}
\label{app:generator_convergence}
This appendix separates the Taylor calculation in the body from convergence of the lattice process. First and second conditional moments identify a candidate generator, but weak convergence also requires control of large increments, tightness, and uniqueness of the limiting martingale problem. We use the standard generator and martingale-problem route \citep{EthierKurtz1986,StroockVaradhan1979}.

\subsection*{Operational triangular array}
For refinement index $m$, let $\Delta u_m\downarrow0$ and let $S_k^{(m)}$ be a time-inhomogeneous Markov chain at $u_k^{(m)}=k\Delta u_m$. Given $S_k^{(m)}=s$, let its increment have kernel $K_m(s,\mathrm dz;u)$. Set
\begin{align}
 S^{(m)}(u)
 &=S_{\lfloor u/\Delta u_m\rfloor}^{(m)},\nonumber\\
 (\mathcal L_mf)(s,u)
 &=\frac{1}{\Delta u_m}
 \int\!\{f(s+z)-f(s)\}K_m(s,\mathrm dz;u).
 \label{eq:app_discrete_generator}
\end{align}
The nearest-neighbour lattice is the special case supported on $\{-\Delta S_m,+\Delta S_m\}$.

\begin{assumption}[Local diffusion approximation]
\label{ass:app_local_diffusion}
On every finite operational horizon $[0,U]$ assume:
\begin{enumerate}[label=(\alph*)]
\item $S^{(m)}(0)\Rightarrow S_0$.
\item $\mu_u(s,u)$ and $D_u(s,u)>0$ are continuous, locally Lipschitz in $s$ uniformly on compact time intervals, and have at most linear growth.
\item Locally uniformly in $(s,u)$,
\begin{align}
 \frac{1}{\Delta u_m}\int zK_m(s,\mathrm dz;u)&\longrightarrow\mu_u(s,u), \label{eq:app_mean_convergence}\\
 \frac{1}{2\Delta u_m}\int z^2K_m(s,\mathrm dz;u)&\longrightarrow D_u(s,u).
 \label{eq:app_moment_convergence}
\end{align}
\item For every $\varepsilon>0$, locally uniformly in $(s,u)$,
\begin{equation}
 \frac{1}{\Delta u_m}\int_{|z|>\varepsilon}z^2K_m(s,\mathrm dz;u)\longrightarrow0.
 \label{eq:app_lindeberg}
\end{equation}
\item Compact containment holds: for every $\eta>0$ there is $R<\infty$ with
\begin{equation}
 \inf_m\Prob\left\{\sup_{0\leq u\leq U}|S^{(m)}(u)|\leq R\right\}\geq1-\eta.
 \label{eq:app_compact_containment}
\end{equation}
\end{enumerate}
\end{assumption}
For a nearest-neighbour chain, condition~\eqref{eq:app_lindeberg} follows from $\Delta S_m\to0$. State dependence in $D_u$ requires a state-dependent local kernel, mesh, or activity; it cannot arise from one homogeneous two-point kernel with fixed spacing and fixed operational step.

\begin{lemma}[Generator convergence]
\label{lem:app_generator_convergence}
Under Assumption~\ref{ass:app_local_diffusion}, for every $f\in C_c^3(\mathbb R)$,
\begin{equation}
 (\mathcal L_mf)(s,u)\longrightarrow(\mathcal L_uf)(s)
 :=\mu_u(s,u) f'(s)+D_u(s,u)f''(s)
 \label{eq:app_generator_convergence}
\end{equation}
locally uniformly on $\mathbb R\times[0,U]$.
\end{lemma}
\begin{proof}
Split the generator integral at $|z|=\varepsilon$. On $|z|\leq\varepsilon$, Taylor's formula gives
\begin{align}
 f(s+z)-f(s)
 &=zf'(s)+\tfrac12z^2f''(s)+R_f(s,z),\nonumber\\
 |R_f(s,z)|&\leq c_f|z|^3.
\end{align}
The first two terms converge by~\eqref{eq:app_mean_convergence} and~\eqref{eq:app_moment_convergence}. The scaled remainder is bounded on compacts by a constant times $\varepsilon(\Delta u_m)^{-1}\int z^2K_m$, and vanishes as $\varepsilon\downarrow0$. The large-increment part vanishes by boundedness of $f$ and~\eqref{eq:app_lindeberg}.
\end{proof}

\begin{theorem}[Operational lattice-to-diffusion convergence]
\label{thm:app_lattice_diffusion_convergence}
Under Assumption~\ref{ass:app_local_diffusion}, suppose the martingale problem for $\mathcal L_u$ with initial law $S_0$ is well posed. Then
\begin{equation}
 S^{(m)}\Rightarrow S\quad\text{in }D([0,U],\mathbb R),
 \label{eq:app_weak_convergence}
\end{equation}
where $S$ is the unique weak solution in law of
\begin{equation}
 \mathrm dS(u)=\mu_u(S(u),u)\,\mathrm du+\sqrt{2D_u(S(u),u)}\,\mathrm dW(u).
 \label{eq:app_limiting_sde}
\end{equation}
\end{theorem}
\begin{proof}
For $f\in C_c^3(\mathbb R)$, the one-step conditional expectation makes
\begin{align}
 &f(S^{(m)}(u))-f(S^{(m)}(0))\nonumber\\
 &\quad-
 \sum_{k<u/\Delta u_m}\!\Delta u_m
 (\mathcal L_mf)(S_k^{(m)},u_k^{(m)})
 \label{eq:app_discrete_martingale}
\end{align}
a martingale up to the final incomplete step. Generator convergence, the large-increment condition, and compact containment give tightness and identify every subsequential limit as a solution of the $\mathcal L_u$ martingale problem. Well-posedness makes the limit unique in law, so the full sequence converges; see \citet{EthierKurtz1986} and \citet{StroockVaradhan1979}.
\end{proof}

\subsection*{What the theorem does and does not prove}
The theorem proves weak convergence of processes. It does not by itself prove pointwise convergence of transition densities. For the limit, the formal adjoint is
\begin{equation}
 (\mathcal L_u^*p)(s,u)
 =-\partial_s\!\left\{\mu_u(s,u)p(s,u)\right\}
 +\partial_{ss}\!\left\{D_u(s,u)p(s,u)\right\}.
 \label{eq:app_adjoint_generator}
\end{equation}
With sufficient regularity and non-degeneracy, the limiting density satisfies the forward and backward equations in the body. Density convergence of the lattice is a stronger local-limit result requiring additional assumptions and is not claimed here.

\section{Deterministic clock projection}
\label{app:deterministic_clock_projection}
This appendix clarifies the deterministic clock step used in the body. It treats a time-only, absolutely continuous clock. This is the case in which operational dynamics can be projected into calendar time by a local rescaling of the generator. State-dependent, endogenous, stochastic, direct-subordinator, and inverse-subordinator clocks are different constructions and are not covered by the theorem below.

\subsection*{Time change of the martingale problem}
Let $S(u)$ solve the operational-time martingale problem for the time-inhomogeneous operator
\begin{align}
 (\mathcal L_u f)(s)
 &=\mu_u(s,u)f'(s)+D_u(s,u)f''(s),\nonumber\\
 &\hspace{24mm} f\in C_c^2(\mathbb R).
 \label{eq:app_clock_operational_generator}
\end{align}
Thus, for each $f\in C_c^2(\mathbb R)$,
\begin{equation}
 M_f(u):=f(S(u))-f(S(0))
 -\int_0^u(\mathcal L_vf)(S(v))\,\mathrm dv
 \label{eq:app_clock_operational_martingale}
\end{equation}
is a martingale in the operational filtration $\{\mathcal F_u\}$.

Let $U:[0,T]\to[0,\bar U]$ be deterministic, nondecreasing, and absolutely continuous, with
\begin{align}
 U(t)&=U(0)+\int_0^t\alpha(\tau)\,\mathrm d\tau,\nonumber\\
 \alpha(t)&\geq0\quad\text{for almost every }t.
 \label{eq:app_clock_absolute_continuity}
\end{align}
Define the calendar-time process and filtration by
\begin{align}
 S_t&:=S(U(t)),\nonumber\\
 \mathcal F_t^U&:=\mathcal F_{U(t)}.
 \label{eq:app_clock_projected_process}
\end{align}
Thus, in this subsection, $S_0=S(U(0))$ denotes the initial projected state.

\begin{theorem}[Deterministic clock projection]
\label{thm:app_deterministic_clock_projection}
Suppose the operational martingale problem for $\mathcal L_u$ is well posed and the integrals below are finite. Under the deterministic clock~\eqref{eq:app_clock_absolute_continuity}, the projected process $(S_t)_{0\leq t\leq T}$ solves the calendar-time martingale problem for
\begin{align}
 (\mathcal A_tf)(s)
 &=\alpha(t)(\mathcal L_{U(t)}f)(s)\nonumber\\
 &=\alpha(t)\mu_u(s,U(t))f'(s)\nonumber\\
 &\quad+\alpha(t)D_u(s,U(t))f''(s).
 \label{eq:app_clock_calendar_generator}
\end{align}
In particular, intervals on which $\alpha(t)=0$ freeze the projected process and give $\mathcal A_t=0$ almost everywhere on those intervals. Strict monotonicity of $U$ is not required for this generator statement.
\end{theorem}

\begin{proof}
Evaluate~\eqref{eq:app_clock_operational_martingale} at the deterministic stopping time $U(t)$. Deterministic time change preserves the martingale property in the filtration $\mathcal F_t^U$. Hence
\begin{equation}
 f(S_t)-f(S_0)
 -\int_{U(0)}^{U(t)}(\mathcal L_vf)(S(v))\,\mathrm dv
 \label{eq:app_clock_changed_martingale}
\end{equation}
is an $\mathcal F_t^U$-martingale. Absolute continuity of $U$ gives the change of variables
\begin{align}
 &\int_{U(0)}^{U(t)}(\mathcal L_vf)(S(v))\,\mathrm dv\nonumber\\
 &\quad=\int_0^t\alpha(\tau)
 (\mathcal L_{U(\tau)}f)(S_\tau)\,\mathrm d\tau.
 \label{eq:app_clock_change_variables}
\end{align}
Substitution into~\eqref{eq:app_clock_changed_martingale} identifies the generator~\eqref{eq:app_clock_calendar_generator}. If $\alpha=0$ on an interval, then $U$ and therefore $S_t$ are constant there.
\end{proof}

Equivalently, the projected martingale problem has the weak SDE representation
\begin{align}
 \mathrm dS_t
 ={}&\alpha(t)\mu_u(S_t,U(t))\,\mathrm dt\nonumber\\
 &+\sqrt{2\alpha(t)D_u(S_t,U(t))}
 \,\mathrm dW_t^U.
 \label{eq:app_clock_calendar_sde}
\end{align}
for a Brownian motion $W^U$ in the weak representation. This follows from the projected generator and does not identify the original operational Brownian path pointwise \citep{EthierKurtz1986}.

\subsection*{Price coordinates and the pricing equation}
Set $X_t=\exp\{S_t\}$ and let $V(x,t)$ be smooth. Since
\begin{equation}
 \partial_s=x\partial_x,
 \qquad
 \partial_{ss}=x^2\partial_{xx}+x\partial_x,
\end{equation}
the projected price-coordinate generator is
\begin{align}
 (\mathcal A_t^XV)(x,t)
 ={}&\alpha(t)D_u(\ln x,U(t))x^2V_{xx}(x,t)\nonumber\\
 &+\alpha(t)\mu_u(\ln x,U(t))xV_x(x,t)\nonumber\\
 &+\alpha(t)D_u(\ln x,U(t))xV_x(x,t).
 \label{eq:app_clock_price_generator}
\end{align}
Under a selected pricing measure $Q$, the discounted cum-dividend price is a martingale when
\begin{equation}
 \alpha(t)\{\mu_u^Q(\ln x,U(t))+D_u^Q(\ln x,U(t))\}=r(t)-q(t).
 \label{eq:app_clock_risk_neutral_drift}
\end{equation}
Define the projected calendar-time local variance by
\begin{equation}
 \sigma_{\mathrm{loc}}^2(x,t)
 :=2\alpha(t)D_u^Q(\ln x,U(t)).
 \label{eq:app_clock_projected_variance}
\end{equation}
For a European payoff $g$, the backward pricing equation is then
\begin{align}
 &V_t+(r-q)xV_x
 +\frac12\sigma_{\mathrm{loc}}^2(x,t)x^2V_{xx}\nonumber\\
 &\hspace{32mm}-rV=0,\qquad V(x,T)=g(x).
 \label{eq:app_clock_pricing_pde}
\end{align}
under the usual conditions ensuring that the conditional-expectation representation and the associated backward problem are well defined.

If $\alpha(t)=0$ while $r(t)-q(t)\neq0$, condition~\eqref{eq:app_clock_risk_neutral_drift} cannot hold for a price process driven only by the frozen operational state. Thus a flat clock interval is mathematically allowed by the projection theorem, but its use in an arbitrage-free stock-pricing model imposes a compatibility condition: either the net carry also vanishes there or the model must contain an additional price or carry mechanism not represented by the frozen operational kernel.

\subsection*{Calendar and operational discounting}
Calendar discounting remains primitive in the pricing problem. If one nevertheless defines an operational short rate $\rho(u)$ and requires equality of discount factors,
\begin{equation}
 \int_t^T r(\tau)\,\mathrm d\tau
 =\int_{U(t)}^{U(T)}\rho(v)\,\mathrm dv,
 \label{eq:app_clock_discount_consistency}
\end{equation}
then, wherever $\alpha(t)>0$ and the relevant derivatives exist,
\begin{equation}
 \rho(U(t))=\frac{r(t)}{\alpha(t)}.
 \label{eq:app_clock_operational_rate}
\end{equation}
This conversion requires local invertibility of the clock. It is not available on a flat interval with $\alpha=0$ unless the accumulated calendar rate also vanishes there. The generator projection therefore needs only a nondecreasing absolutely continuous clock, whereas an operational-rate representation needs the stronger condition $\alpha>0$ on the interval being inverted.

\subsection*{Boundary of the deterministic result}
The theorem establishes only the deterministic time-only projection
\begin{equation}
 \mathcal A_t=\alpha(t)\mathcal L_{U(t)}.
\end{equation}
When a calendar-time pricing model is constructed from an operational kernel, a clock, and a pricing rule, the more qualified structural dependence is
\begin{equation}
 (K_{\Delta u},U,Q)\longmapsto\mathcal A_t^Q.
 \label{eq:app_joint_kernel_clock_measure}
\end{equation}
This notation records joint dependence; it does not assert that clock projection and measure selection commute. Outside the deterministic local-diffusion setting, a change of measure may alter the operational kernel, a jump compensator, the clock law, or their dependence. The simpler sequential notation used in the body remains appropriate for the deterministic construction developed there.

A state-dependent rate requires a path-dependent additive functional. For a time-homogeneous operational semigroup, direct independent subordination gives the classical nonlocal subordinated generator; a time-inhomogeneous operational model instead requires an augmented-state or evolution-family formulation. Inverse or renewal clocks, and stochastic or endogenous clocks, are also separate constructions and are not covered by this local deterministic result; they may respectively produce memory or additional unspanned risk.

\section{Clock--variance equivalence}
\label{app:clock_variance_equivalence}
This appendix formalises the deterministic clock--variance identification point used in the body. The result is deliberately narrow. It concerns operational models which, after deterministic clock projection and selection of a pricing measure, induce the same one-dimensional calendar-time local generator. Within that class, the operational variance and the activity clock are not separately identified by European option prices. This is an observational-equivalence statement, not a claim that arbitrary stochastic-clock or hidden-state models are equivalent.

\subsection*{Projected pricing specifications}
For $i\in\{1,2\}$, let
\begin{align}
 U_i(t)&=U_i(0)+\int_0^t\alpha_i(\tau)\,\mathrm d\tau,\nonumber\\
 \alpha_i(t)&>0\quad\text{for almost every }t.
 \label{eq:app_equiv_clocks}
\end{align}
be deterministic, absolutely continuous clocks. Under a selected pricing measure $Q_i$, let the operational log-price generator be
\begin{equation}
 (\mathcal L_{i,u}^{Q_i}f)(s)
 =\mu_i^{Q_i}(s,u)f'(s)+D_i^{Q_i}(s,u)f''(s).
 \label{eq:app_equiv_operational_generators}
\end{equation}
By Theorem~\ref{thm:app_deterministic_clock_projection}, the projected calendar generator is
\begin{equation}
 (\mathcal A_{i,t}^{Q_i}f)(s)
 =b_i^Q(s,t)f'(s)+a_i^Q(s,t)f''(s),
 \label{eq:app_equiv_calendar_generators}
\end{equation}
where
\begin{align}
 b_i^Q(s,t)&=\alpha_i(t)\mu_i^{Q_i}(s,U_i(t)),\nonumber\\
 a_i^Q(s,t)&=\alpha_i(t)D_i^{Q_i}(s,U_i(t)).
 \label{eq:app_equiv_projected_coefficients}
\end{align}
In price coordinates $x=e^s$, the martingale restriction is
\begin{equation}
 b_i^Q(s,t)+a_i^Q(s,t)=r(t)-q(t).
 \label{eq:app_equiv_martingale_restriction}
\end{equation}

\begin{definition}[Calendar-generator equivalence]
\label{def:app_calendar_generator_equivalence}
Two deterministic operational pricing specifications are calendar-generator equivalent on $[0,T]$ if they have the same initial calendar state, the same rates and dividends, and
\begin{equation}
 b_1^Q(s,t)=b_2^Q(s,t),
 \qquad
 a_1^Q(s,t)=a_2^Q(s,t)
 \label{eq:app_equiv_definition}
\end{equation}
for almost every $t$ and all $s$ in the common state domain.
\end{definition}

\begin{theorem}[Clock--variance equivalence]
\label{thm:app_clock_variance_equivalence}
Suppose the calendar martingale problem associated with the common coefficients in~\eqref{eq:app_equiv_definition} is well posed for every relevant starting pair $(s,t)$. If two deterministic operational specifications are calendar-generator equivalent, then their projected log-price Markov families have the same law on $D([t,T],\mathbb R)$ from each such starting pair. Consequently, for every integrable European payoff $g$ and every $0\leq t\leq T$, their prices agree:
\begin{align}
 V_i(x,t)
 &=B_r(t,T)\,
 \mathbb E^{Q_i}\!\left[g(X_T^{(i)})\mid X_t^{(i)}=x\right],\nonumber\\
 V_1(x,t)&=V_2(x,t).
 \label{eq:app_equiv_european_prices}
\end{align}
More generally, any integrable calendar-time path functional has the same value under the common projected law and common discounting convention.
\end{theorem}

\begin{proof}
By Definition~\ref{def:app_calendar_generator_equivalence}, both projected processes solve the same time-inhomogeneous martingale problem from each relevant $(s,t)$. Well-posedness gives a common Markov family and uniqueness in law on path space. Applying the same discounted conditional-expectation functional gives the price equality; equality for integrable path functionals follows from equality of path laws.
\end{proof}

The path-law statement is important for interpretation. A deterministic refactorisation of one and the same calendar generator cannot by itself produce different barrier prices or different calendar-time hedging distributions. Such differences require a weaker matching criterion, such as equality only of terminal marginals, equality only in a diffusion limit, or an enlarged random-clock or hidden-state model. This distinction is used in the applied example developed in \ref{app:unspanned_clock_risk}.

\subsection*{Clock--variance factorisation}
Under the martingale restriction~\eqref{eq:app_equiv_martingale_restriction}, equality of projected local variances is sufficient for calendar-generator equivalence. In particular, suppose
\begin{equation}
 \alpha_1(t)D_1^{Q_1}(s,U_1(t))
 =\alpha_2(t)D_2^{Q_2}(s,U_2(t))
 =\frac12\sigma_{\mathrm{loc}}^2(e^s,t).
 \label{eq:app_equiv_variance_product}
\end{equation}
Then the risk-neutral log drifts are fixed by
\begin{equation}
 \alpha_i(t)\mu_i^{Q_i}(s,U_i(t))
 =r(t)-q(t)-\frac12\sigma_{\mathrm{loc}}^2(e^s,t),
 \label{eq:app_equiv_log_drift}
\end{equation}
and the projected generators coincide.

This generates an equivalence class of operational representations. Let $U$ be any strictly increasing, absolutely continuous clock with activity $\alpha>0$. For a prescribed projected coefficient $a^Q(s,t)$, define along the clock image
\begin{align}
 D_U^Q(s,U(t))&:=\frac{a^Q(s,t)}{\alpha(t)},\nonumber\\
 \mu_U^Q(s,U(t))&:=
 \frac{r(t)-q(t)-a^Q(s,t)}{\alpha(t)}.
 \label{eq:app_equiv_factorisation}
\end{align}
Every admissible choice of $U$ produces the same calendar generator after projection. The option model therefore identifies the projected coefficient $a^Q$, while the decomposition into $U$, $D_U^Q$, and $\mu_U^Q$ requires additional operational information or structural restrictions.

\begin{corollary}[Non-identifiability from the European surface]
\label{cor:app_option_surface_nonidentifiability}
Assume a sufficiently smooth, arbitrage-free European call surface is generated by a one-dimensional deterministic local-volatility diffusion with known rates and dividends, and that the Dupire inversion is well defined. Then the call surface identifies the projected risk-neutral local variance
\begin{equation}
 \sigma_{\mathrm{loc}}^2(x,t)
 =2\alpha(t)D_u^Q(\ln x,U(t)),
 \label{eq:app_equiv_identified_product}
\end{equation}
but does not separately identify $\alpha(t)$ and $D_u^Q(\ln x,U(t))$ within the deterministic operational class.
\end{corollary}

\begin{proof}
Under the stated local-volatility assumptions, the call surface determines the calendar-time local-variance coefficient through the Dupire relation in~\eqref{eq:dupire_type_relation} \citep{Dupire1994,Gatheral2006}. Equation~\eqref{eq:app_equiv_factorisation} constructs infinitely many positive-clock factorisations of the same coefficient. By Theorem~\ref{thm:app_clock_variance_equivalence}, these factorisations yield the same projected pricing law and hence the same European surface.
\end{proof}

\subsection*{What additional information can identify the components?}
Separate identification requires information not contained in the vanilla option surface alone. Examples include an observed event-count or transaction-intensity clock, a structural reaction-boundary model fixing $D_u^Q$, restrictions linking physical and pricing kernels, or claims whose settlement or payoff is written directly on operational activity. \ref{app:localvolfromlinearlob} provides one possible structural restriction by relating a finite-scale boundary variance to liquidity slope, forcing intensity, memory, resilience, and measurement scale. That relation does not remove the need to specify the calendar clock.

The qualification ``within the deterministic operational class'' is essential. A stochastic clock can alter joint laws, introduce a separate risk premium, or produce a non-Markovian calendar projection. A general hidden-state model may reproduce the same European marginals without sharing the same calendar generator. The corollary does not identify such models with the deterministic local-volatility projection; it says only that European prices cannot split a projected deterministic variance coefficient into a unique activity factor and a unique operational variance factor. 

\subsection*{Relation to the synthesis in the body}
The equivalence result gives a precise form to the qualified dependence introduced in \ref{app:deterministic_clock_projection}:
\begin{equation}
 (K_{\Delta u},U,Q)\longmapsto\mathcal A_t^Q
 \longmapsto\text{pricing equations}.
\end{equation}
Different operational kernels, deterministic clocks, and compatible pricing specifications can map to the same $\mathcal A_t^Q$. Once that happens, European pricing cannot recover which operational representation generated the projected coefficient. The reaction-boundary interpretation remains useful, but it is a structural interpretation requiring information beyond the option surface. This is the identification content of keeping the kernel, clock, and pricing measure separate.

\section{Unequal finite-mesh paths}
\label{app:finite_mesh_non_equivalence}
The equivalence theorem in \ref{app:clock_variance_equivalence} uses equality of the full projected calendar generator and therefore equality of path laws. It cannot produce different values for path-dependent claims. This appendix uses a deliberately weaker matching criterion: two arbitrage-free finite-mesh models have the same terminal distribution and therefore the same prices for every European payoff at one maturity, but distribute activity differently across the intermediate date. A discretely monitored claim then separates them.

\subsection*{Two-period construction}
Take dates $t_0=0<t_1<t_2=T$, zero rates and dividends, and an initial forward or discounted stock value $F_0>2\varepsilon$ for some $\varepsilon>0$. The positivity condition keeps every node strictly positive. Consider two models under their stated pricing measures.

In Model A, activity is spread across both periods. Let $\xi_1$ and $\xi_2$ be independent random variables with
\begin{equation}
 \Prob(\xi_j=1)=\Prob(\xi_j=-1)=\frac12,
 \qquad j=1,2,
 \label{eq:app_finite_model_a_increments}
\end{equation}
and define
\begin{equation}
 F_{t_1}^{A}=F_0+\varepsilon\xi_1,
 \qquad
 F_T^{A}=F_0+\varepsilon(\xi_1+\xi_2).
 \label{eq:app_finite_model_a}
\end{equation}
This is the usual two-step symmetric nearest-neighbour martingale.

In Model B, the first calendar interval contains no price activity and the same terminal law is applied in the second interval. Set
\begin{equation}
 F_{t_1}^{B}=F_0,
 \label{eq:app_finite_model_b_first}
\end{equation}
and let
\begin{equation}
 F_T^{B}=
 \begin{cases}
 F_0+2\varepsilon,&\text{with probability }1/4,\\
 F_0,&\text{with probability }1/2,\\
 F_0-2\varepsilon,&\text{with probability }1/4.
 \end{cases}
 \label{eq:app_finite_model_b_terminal}
\end{equation}
Model B is a martingale because $\mathbb E[F_T^B\mid F_{t_1}^B]=F_0=F_{t_1}^B$. At finite mesh its second step is a three-point kernel rather than a nearest-neighbour kernel. This is intentional: it isolates the effect of moving the same terminal uncertainty into a different calendar interval.

\begin{proposition}[Terminal equivalence and pathwise non-equivalence]
\label{prop:app_finite_terminal_equivalence}
Models A and B have the same terminal distribution,
\begin{align}
 \Prob(F_T=F_0+2\varepsilon)&=\frac14,\nonumber\\
 \Prob(F_T=F_0)&=\frac12,\nonumber\\
 \Prob(F_T=F_0-2\varepsilon)&=\frac14.
 \label{eq:app_finite_common_terminal_law}
\end{align}
Hence they assign the same time-zero price to every integrable European payoff $g(F_T)$. However, for the discretely monitored barrier
\begin{equation}
 H:=\mathbf 1_{\left\{
 \max(F_{t_1},F_T)\geq F_0+\varepsilon
 \right\}}.
 \label{eq:app_finite_one_touch}
\end{equation}
the prices differ:
\begin{equation}
 \mathbb E[H^A]=\frac12,
 \qquad
 \mathbb E[H^B]=\frac14.
 \label{eq:app_finite_one_touch_values}
\end{equation}
\end{proposition}

\begin{proof}
For Model A, the sum $\xi_1+\xi_2$ takes the values $2,0,-2$ with probabilities $1/4,1/2,1/4$, which gives~\eqref{eq:app_finite_common_terminal_law}; Model B was constructed with that law. Therefore
\begin{equation}
 \mathbb E[g(F_T^A)]=\mathbb E[g(F_T^B)]
 \label{eq:app_finite_european_equality}
\end{equation}
for every integrable $g$.

In Model A, the upper barrier is reached at $t_1$ exactly when $\xi_1=1$. If $\xi_1=-1$, the largest possible terminal value is $F_0$, so the barrier cannot subsequently be reached. Thus $\Prob(H^A=1)=1/2$. In Model B, the process remains at $F_0$ at $t_1$ and reaches the barrier only in the upper terminal state $F_0+2\varepsilon$, which has probability $1/4$.
\end{proof}

Both models also have the same terminal mean and variance,
\begin{equation}
 \mathbb E[F_T]=F_0,
 \qquad
 \operatorname{Var}(F_T)=2\varepsilon^2.
 \label{eq:app_finite_terminal_moments}
\end{equation}
but place activity differently: Model A reveals uncertainty in both periods, whereas Model B concentrates it in the second. This is not calendar-generator equivalence because the intermediate transition laws and path laws differ. It shows compactly that terminal-moment matching is weaker than terminal-law matching, terminal-law matching fixes European prices only at that maturity, and full generator equivalence is needed to fix the projected path law. A second maturity or a monitored claim distinguishes the models. Likewise, operational kernels with a common continuum limit may retain different finite-mesh prices or hedging errors; the limiting PDE is a common approximation, not an identity between the event systems.

\section{Unspanned clock risk}
\label{app:unspanned_clock_risk}
This appendix gives a minimal incomplete-market example. The point is not that every random clock makes a market incomplete. The point is that a clock state which affects a claim but is not distinguished by traded-asset payoffs is unspanned. No-arbitrage then restricts, but does not uniquely determine, the pricing law of that state. This is the standard relation between unattainable claims and non-unique equivalent martingale measures in incomplete markets \citep{BinghamKiesel2004}.

\subsection*{A one-period market with a clock label}
Take dates $0$ and $T$, zero interest, and a strictly positive stock with initial value $X_0>\delta>0$. The terminal state space is
\begin{equation}
 \Omega=\{(U,F),(U,S),(D,F),(D,S)\}.
 \label{eq:app_clockrisk_state_space}
\end{equation}
The first coordinate records an up or down stock move. The second is a clock label: $F$ denotes a fast-activity state and $S$ a slow-activity state. Assume the physical measure $P$ assigns strictly positive probability to every state.

The traded money-market account is constant and the stock payoff is
\begin{align}
 X_T(U,F)=X_T(U,S)&=X_0+\delta,\nonumber\\
 X_T(D,F)=X_T(D,S)&=X_0-\delta.
 \label{eq:app_clockrisk_stock_payoff}
\end{align}
Thus the stock distinguishes price direction but not the clock label within a price state.

For every $\theta\in(0,1)$, define a probability measure $Q_\theta$ by
\begin{align}
 Q_\theta(U,F)&=\frac{\theta}{2},&
 Q_\theta(U,S)&=\frac{1-\theta}{2},\nonumber\\
 Q_\theta(D,F)&=\frac{\theta}{2},&
 Q_\theta(D,S)&=\frac{1-\theta}{2}.
 \label{eq:app_clockrisk_measure_family}
\end{align}
Every $Q_\theta$ is equivalent to $P$ because all four probabilities are strictly positive.

\begin{proposition}[A family of equivalent martingale measures]
\label{prop:app_clockrisk_emm_family}
Every measure $Q_\theta$, $\theta\in(0,1)$, is an equivalent martingale measure for the traded stock and money-market account. Hence the market is arbitrage free but incomplete.
\end{proposition}

\begin{proof}
Under $Q_\theta$, the total probability of an up move is
\begin{equation}
 Q_\theta(U,F)+Q_\theta(U,S)=\frac12,
\end{equation}
and the total probability of a down move is also $1/2$. Consequently,
\begin{equation}
 \mathbb E^{Q_\theta}[X_T]
 =\frac12(X_0+\delta)+\frac12(X_0-\delta)=X_0.
 \label{eq:app_clockrisk_stock_martingale}
\end{equation}
The bank account is constant, so discounted traded prices are martingales under every $Q_\theta$. Existence of these equivalent martingale measures gives absence of arbitrage in this finite-state model. Their non-uniqueness, or equivalently the unattainability shown below, gives incompleteness.
\end{proof}

\subsection*{A clock-sensitive claim}
Consider the claim
\begin{equation}
 H_F:=\mathbf 1_{\{F\}},
 \label{eq:app_clockrisk_claim}
\end{equation}
which pays one unit if the terminal clock state is fast and zero if it is slow. Its expectation under $Q_\theta$ is
\begin{equation}
 \pi_\theta(H_F)
 =\mathbb E^{Q_\theta}[H_F]
 =Q_\theta(U,F)+Q_\theta(D,F)=\theta.
 \label{eq:app_clockrisk_claim_price}
\end{equation}
No-arbitrage therefore does not select a unique value for the clock-sensitive claim.

\begin{lemma}[The clock claim is not attainable]
\label{lem:app_clockrisk_unattainable}
The claim $H_F$ cannot be replicated by a static portfolio in the stock and money-market account.
\end{lemma}

\begin{proof}
A one-period portfolio holding $a$ units of cash and $b$ units of stock has terminal payoff
\begin{equation}
 a+bX_T.
 \label{eq:app_clockrisk_portfolio_payoff}
\end{equation}
This payoff has the same value in $(U,F)$ and $(U,S)$ because the stock payoff is $X_0+\delta$ in both states. It also has the same value in $(D,F)$ and $(D,S)$. By contrast, $H_F$ equals one in the fast state and zero in the slow state, including within a fixed price direction. No choice of $a$ and $b$ can reproduce it.
\end{proof}

The interval in~\eqref{eq:app_clockrisk_claim_price} records values under strictly positive equivalent measures; its endpoints correspond to non-equivalent limiting measures and are not asserted to be complete superhedging bounds. The central point is only that stock trading spans the up/down direction but not the fast/slow label within either price state. The free parameter $\theta$ is therefore a clock-risk pricing choice.

Randomness alone is not sufficient. Clock risk creates the displayed incompleteness only when the clock state is economically relevant to a claim, is not spanned by traded payoffs, and remains unrestricted by the martingale conditions. A traded clock-linked claim priced at $c$ would impose $\E^Q[H_F]=c$ and fix $\theta=c$ within this family, although full completion of the four-state market would require enough independent securities to span all state payoffs. In an operational-time model the label may represent event activity, renewal state, or latent liquidity; the example shows why the stock drift restriction need not select its risk premium.

\end{document}